\documentclass[11pt]{article}

\usepackage{amsmath}
\usepackage{amsfonts}
\usepackage{amsthm}
\usepackage{enumitem}
\usepackage[usenames]{color}
\usepackage{setspace}
\usepackage{mathrsfs}
\usepackage{amssymb}
\usepackage{mathtools}
\usepackage{tabularray}
\usepackage{sgame}
\usepackage{tikz}
\usetikzlibrary{arrows}
\usetikzlibrary{arrows.meta}
\usepackage{multirow}
\usepackage[round]{natbib}
\usepackage{graphics}
\usepackage{todonotes}

\theoremstyle{definition}

\theoremstyle{plain}

\theoremstyle{definition}

\theoremstyle{definition}

\theoremstyle{plain}

\theoremstyle{definition}
\newtheorem{definition}{Definition}
\theoremstyle{definition}

\theoremstyle{plain}
\newtheorem{theorem}{Theorem}
\theoremstyle{plain}

\theoremstyle{plain}
\newtheorem{lemma}{Lemma}
\theoremstyle{plain}
\newtheorem{corollary}{Corollary}
\theoremstyle{definition}

\theoremstyle{remark}

\theoremstyle{definition}

\theoremstyle{definition}

\theoremstyle{definition}

\theoremstyle{plain}
\newtheorem*{thm}{Theorem}
\theoremstyle{plain}

\theoremstyle{plain}
\newtheorem{claim}{Claim}

\newcommand{\red}[1]{{\color{red}#1}}
\newcommand{\blue}[1]{{\color{blue}#1}}
\newcommand{\green}[1]{{\color{green}#1}}

\DeclareMathOperator*{\argmax}{argmax}

\newcommand{\set}[1]{\left\{#1\right\}}

\newcommand{\strat}[1]{\ensuremath{s_{#1}}}

\newcommand{\stratBold}[1]{\boldsymbol{s}_{#1}}
\newcommand{\stratSpace}{\boldsymbol{S}}

\newcommand{\nA}{n_{A}}
\newcommand{\nB}{n_{B}}

\newcommand{\WelfareProblem}{{\sc Welfare Maximisation}[$\Fcal$]}
\newcommand{\PotentialProblem}{{\sc Potential Maximisation}[$\Fcal$]}

\newcommand{\pay}[1]{\gamma_{#1}}

\newcommand{\NP}{\textsf{NP}}

\newcommand{\one}{a}
\newcommand{\two}{b}

\newcommand{\sbf}[1]{\boldsymbol{s}_{#1}}
\newcommand{\welf}[1]{\mathcal{W}{#1}}

\newcommand{\pot}{\Phi}
\newcommand{\Real}{\mathbb{R}}

\newcommand{\Fcal}{\mathcal{F}}
\newcommand{\Gcal}{\mathcal{G}}

\newcommand{\Ocal}{\mathcal{O}}

\newcommand{\MWOP}{Maximum Weighted Orgraph Partition}

\usepackage[pdftex, citecolor=blue, colorlinks]{hyperref}
\usepackage[margin=3cm]{geometry}

\begin{document}

\title{ \huge Some coordination problems are harder than others}
\author{Argyrios Deligkas\footnote{Computer Science Department, Royal Holloway University of London.} \and Eduard Eiben\footnote{Computer Science Department, Royal Holloway University of London.} \and Gregory Gutin\footnote{Computer Science Department, Royal Holloway University of London.} \and Philip R.\ Neary\footnote{Economics Department, Royal Holloway University of London.} \and Anders Yeo\footnote{IMADA, University of Southern Denmark.} $\,^{,}$\footnote{Department of Mathematics, University of Johannesburg.}}


\date{\today}

\maketitle


\linespread{1.3}


\begin{abstract}
\noindent
In order to coordinate players in a game must first identify a target pattern of behaviour.
In this paper we investigate the difficulty of identifying prominent outcomes in two kinds of binary action coordination problems in social networks:\ pure coordination games and anti-coordination games.
For both environments, we determine the computational complexity of finding a strategy profile that
(i) maximises welfare,
(ii) maximises welfare subject to being an equilibrium,
and (iii) maximises potential.
We show that the complexity of these objectives can vary with the type of coordination problem.
Objectives (i) and (iii) are tractable problems in pure coordination games,
but for anti-coordination games are {\sf NP}-hard.
Objective (ii), finding the ``best'' Nash equilibrium, is {\sf NP}-hard for both. 
Our results support the idea that environments in which actions are strategic complements (e.g., technology adoption) facilitate successful coordination more readily than those in which actions are strategic substitutes (e.g., public good provision).

\vspace{.5in}


\end{abstract}



\newpage


\setstretch{1.3}

\section{Introduction}\label{sec:Intro}


\subsection{Motivation}

What exactly does it mean to say that a group of economic agents successfully coordinated?
Despite the frequency with which issues related to coordination are discussed -{}- just open any game theory journal -{}- finding precise (i.e., mathematical) definitions is not so easy.
Let us spare the reader any suspense.
We do not propose a formal definition of what it means for economic agents ``to coordinate'' in this paper.
That being said, we contend that a key component of any such definition will be the identification of a strategy profile (at least one) with certain properties.
While identification is not the whole story, it raises an obvious question:\ given a coordination problem, is computing an outcome with certain properties always tractable?
If yes, then great.
But if no, is it sensible to suppose that economic agents, even those blessed with other-worldly levels of rationality as is commonly presumed, can identify such an outcome in a reasonable time frame?
And if identifying a particular pattern of behaviour is computationally challenging, ought it be viewed as a plausible prediction?
We focus on the two fundamental types of binary-action coordination problem in social networks:\ pure coordination games, in which it is beneficial to act the same as others (actions are {\it strategic complements}), and anti-coordination games, where it pays to differentiate your behaviour from that of your neighbours (actions are {\it strategic substitutes}).
For both kinds of coordination problem, we consider the computational complexity of identifying the following three commonly-studied outcomes:
\begin{enumerate}[label=(\roman*)]
\item
a strategy profile that maximises welfare,
\item
a welfare maximising Nash equilibrium, and
\item
a strategy profile that maximises the potential \citep{ShapleyMonderer:1996:GEB}.
\end{enumerate}

The profile in (i) is that which a social planner might be expected to favour.
Unfortunately for a planner, such an outcome lacks the glue of equilibrium and so is likely to unravel.
For this reason we also consider the profile in (ii).
This describes a pattern of behaviour that a planner might view as optimal subject to a stability requirement.
The profile in (iii), when it exists, describes the robust outcome in games of incomplete information \citep{Ui:2001:E,OyamaTakahashi:2020:E} and that deemed the most likely long run outcome in models of stochastic evolution \citep{Blume:1993:GEB, Young:2001:}.\footnote{A potential maximising outcome must be a pure strategy Nash equilibrium. It also possesses other nice properties. For example, potential maximisers are uniquely absorbing \citep{HofbauerSorger:1999:JET}, select risk-dominant outcomes in games on fixed networks \citep{Peski:2010:JET} and games on random networks \citep{Peski:2021:arXiv}, and can be useful for equilibrium refinements \citep{Carbonell-NicolauMcLean:2014:TE}. Potential games appear frequently in applied work -- to study the effects of price discrimination policies in oligopolies \citep{ArmstrongVickers:2001:RAND}, the impact of uncertainty on technology adoption \citep{OstrovskySchwarz:2005:RAND}, and issues of collective action \citep{MyattWallace:2009:EJ} -- and the potential maximiser is often invoked as an ad hoc equilibrium selection device.}

Given two kinds of coordination problem and three strategy profiles of interest, our goal is to classify the complexity of six different objects.
The main results, an overview of which is provided in Table~\ref{table:classification} in Section~\ref{ssec:resultsSummary}, may be summarised, in very broad strokes, as saying that the complexity of identifying certain outcomes in pure coordination games is tractable whereas for anti-coordination games it is not.



We are not the first to consider the complexity of finding outcomes in strategic games.\footnote{A quick review of the literature is provided in Section~\ref{ssec:literatureComplexity}.}
We focus on this issue in coordination problems for two reasons.
The first is to highlight a potential discrepancy in how difficult a game is for players to understand as compared with how difficult it is to ``solve'' conditional on understanding.
That is, just because a game is uncomplicated does not guarantee that playing it, let alone solving it, must be.
While this distinction is understood for some environments (e.g., board games), at least in our opinion, the lines become blurred when it comes to coordination problems.
Perhaps nowhere is this more apparent than in the classic work of Thomas Schelling \citep{Schelling:1960:}, who goes so far as to describe certain equilibria as {\it focal}.\footnote{Not only can the set of candidate outcomes be identified by one and all, but each individual notes that there is something inherently salient about candidate outcome, correctly conjectures that others have also noted this salience, correctly conjectures that others have correctly conjectured that others have also noted this salience, and so on. With everyone on the same page, each player then plays their part and the focal equilibrium emerges.}
A discussion of which equilibrium the population is most likely to end up at only holds water if the set of candidate outcomes is in some sense obvious.
And as we will show, this need not always be the case.
The two kinds of coordination problem that we consider in this paper (see Definitions~\ref{def:pure} and \ref{def:anti}) are uncomplicated games, and yet computing certain outcomes in a reasonable time frame is not guaranteed.\footnote{This stands in contrast to many existing intractability results on equilibrium computation since often these rely on constructing games that are difficult to ascribe any real world meaning to. One example that features prominently in this paper is the language game of \cite{Neary:2011:MGG,Neary:2012:GEB}. This simple coordination game has been tested experimentally \citep{LimNeary:2016:GEB} and has even been used to introduce game theory to high school students at university open days.
And yet, despite how undemanding it is both to explain and to understand, it turns out that identifying the welfare optimal Nash equilibrium is computationally intractable.}




The second reason for our focus on coordination problems is their importance.
Put simply, it is difficult to overstate the role that coordination problems play in everyday life. 
The adoption of technological standards \citep{KatzShapiro:1985:AER,FarrellSaloner:1985:RJE,Arthur:1989:EJ}, the setting of macroeconomic policy \citep{CooperJohn:1988:QJE}, determining the drivers of bank runs \citep{DiamondDybvig:1983:JPE}, and the provision of public goods \citep{Hellwig:2003:RES}, are but a few of the many economic systems wherein coordination is key.
There are immediate benefits to one and all when the right collections of mutually consistent actions appear, but because population behaviour is decentralised the emergence of these ``good'' outcomes is far from guaranteed.
Given this, an understanding of what renders some coordination problems easier to solve than others would be of genuine practical importance.
The ability to facilitate the emergence of a ``good'' outcomes could lead to a dramatic reductions in the cost of coordination failure and in turn lead to potentially enormous gains in welfare.
Our results indicate that, at least when viewed through an algorithmic lens, games of anti-coordination appear more difficult to resolve than pure coordination games.
Whether and how these findings carry over to real life economic settings is a question for experimental economists, but it is our hope that our theoretical results will prove complementary to the large experimental literature devoted to understanding coordination.\footnote{While the potential discrepancy between pure coordination games and anti-coordination games has not been considered, certain features of coordination problems have been documented as relevant. Clearly the relative attractiveness of the outcomes will matter. Consequently, many of the experimental studies have considered environments that have the flavour of Rousseau's stag hunt wherein a tension exists between a risky equilibrium that is optimal and another equilibrium that is suboptimal but safe. \cite{Van-HuyckBattalio:1990:AER} observe that size of population is an important factor in such settings:\ a robust finding is that the smaller the group the higher average payoffs achieved. (\cite{Crawford:1991:GEB} provides an evolutionary explanation for these findings.) Perhaps closer in spirit to our work, \cite{KearnsSuri:2006:S} and \cite{McCubbinsPaturi:2009:} consider colouring problems, i.e., anti-coordination games, on a variety of different network structures and conclude that certain topologies promote successful colouring more readily than others.}

The remainder of the introduction presents an example of the two kinds of coordination problem that we consider followed by a brief overview of our results.
Section~\ref{sec:environment} formally defines the environments and the strategy profiles of interest.
Section~\ref{sec:results} states the main results and sketches how one might, for example, compute the potential maximiser of a pure coordination game.
To our knowledge, this is the first example where the potential maximiser is non trivial and yet can by computed efficiently.
All proofs are in the Appendix.

\subsection{Pure coordination games and anti-coordination games}\label{ssec:examples}


We consider binary action polymatrix games (see Definition~\ref{def:poly} for the formal statement), a subclass of network games first introduced in \cite{Janovskaya:1968:}.
Games in this class are easy to describe:\
a population of players live on a graph where vertices represent players and edges represent pairwise connections;
along every edge is a binary action two-player game;
every player's utility is the sum of payoffs earned with each neighbour where the same action must be used with each.\footnote{Without being specifically referred to as such, polymatrix games appear throughout economics. For example, the stage game of many classic models in stochastic evolutionary game theory is a binary action polymatrix game \citep{KandoriMailath:1993:E,Ellison:1993:E,Morris:2000:RES,Peski:2010:JET,KreindlerYoung:2013:GEB,NewtonSercombe:2020:GEB}.}

We will consider pure coordination polymatrix games (Definition~\ref{def:pure}), those where every two player game is a game of pure coordination (actions are strategic complements), and anti-coordination polymatrix games (Definition~\ref{def:anti}), those in which every two player game is one of anti-coordination (actions are strategic substitutes).

While this paper and its results pertain to {\it all} binary action pure/~anti-coordination polymatrix games, let us now consider an example of each kind.
Suppose there are two societies that are identical both in their size and in their network structure.
That is, for every player in one society, there is a ``mirror image'' player in the second society who has exactly the same neighbourhood composition.
The societies are differentiated only by the manner in which neighbouring players interact.
Specifically, in the first society, each player $i$ interacts with every neighbour $j$ through $G_{ij}^c$ below in Figure~\ref{fig:binaryGames},
while in the second society, player $i$ interacts with each neighbour $j$ via $G_{ij}^{s}$ in Figure~\ref{fig:binaryGames}.
We assume that $\pay{i}, \pay{j} \in [0,1]$ in both $G_{ij}^c$ and $G_{ij}^s$.
(The superscript ``c'' in $G_{ij}^{c}$ stands for complement while the superscript ``s'' in $G_{ij}^{s}$ stands for substitute.)

\begin{figure}[htp!]
\hspace*{\fill}\begin{game}{2}{2}[${i}$][${j}$]
& $a$ & $b$\\
$a$ &$\pay{i},\pay{j}$ & $0,0$\\
$b$ &$0,0$ &1$-\pay{i}$, 1$-\pay{j}$
\end{game}\hspace*{\fill}
\begin{game}{2}{2}[${i}$][${j}$]
& $a$ & $b$\\
$a$ &0,0 & $\pay{i}$, $1-\pay{j}$\\
$b$ & $1-\pay{i}$, $\pay{j}$ &$0,0$
\end{game}\hspace*{\fill}
\caption[]{$G_{ij}^c$ and $G_{ij}^{s}$.}
\label{fig:binaryGames}
\end{figure}

Games of the form $G_{ij}^c$ have two strict equilibria, both profiles along the main diagonal, while those of the form $G_{ij}^s$ have as strict equilibria the two profiles located on the off diagonal.\footnote{We normalise payoffs at all non-equilibrium outcomes to be zero. We do this for expositional purposes only; our main results do not require it.}



Observe that the two environments above are extremely close in spirit.
In each, a player's utility depends only on her own choice of strategy and the number of neighbours playing each strategy (but, importantly, {\it not} on the identity of the neighbouring players playing each strategy).\footnote{\cite{Blonski:1999:GEB,Blonski:2000:JME} considers {\it anonymous games}. Both environments we consider are ``locally'' anonymous.}
Building on this, let us zoom in on two players located at the corresponding vertex of either network and let us consider the strategic environment faced by each.
For this pair of players, optimising behaviour is described by an almost identical {\it threshold rule}.\footnote{The pure coordination polymatrix game where each pairwise interaction has the structure of $G_{ij}^c$ is the {\it threshold model} of \cite{NearyNewton:2017:JMID}. This moniker was chosen as it is the polymatrix game variant of Granovetter's classic threshold model from sociology \citep{Granovetter:1978:AJS}.}
In the pure coordination environment, player $i$'s best response is:\ choose action $b$ whenever at least fraction $\pay{i}$ of neighbours choose action $b$, and choose action $a$ otherwise. 
In the anti-coordination environment, player $i$'s best response is:\ choose action $b$ whenever at least fraction $\pay{i}$ of neighbours choose action $a$, and choose action $b$ otherwise. 


Yet despite how similar optimal behaviour is in both environments, it turns out that one environment is more ``complex'' than the other.
Here is an example of one such discrepancy:\ both $G_{ij}^c$ and $G_{ij}^{s}$ are potential games, and when every two player game of a polymatrix game is a potential game, then the polymatrix game is too.
A well-studied object in potential games is the potential maximising outcome.
It turns out that computing the potential maximising outcome in the pure coordination polymatrix game is a tractable problem, whereas computing the analogous outcome in the anti-coordination environment is not (formally, it is {\sf NP}-hard).

\subsection{Summary of main results}\label{ssec:resultsSummary}

A classification of the tractability of each of the six objectives (three kinds of strategy profile for each of the two kinds of coordination problem) is given in Table~\ref{table:classification}.

\begin{table}[htp]
\caption{Complexity of each objective.}
\begin{center}
\begin{tabular}{|c||c|c|}
\hline
& pure coordination & anti-coordination \\
\hline
\hline
\multirow{3}{*}{maximising welfare} & tractable &  {\sf NP}-hard \\ 
& Corollary~\ref{cor:pure-coord} & Corollary~\ref{cor:anti-coord} \\ 
& & \cite{Bramoulle:2007:GEB} and \cite{CaiDaskalakis:2011:} \\
\hline 
\multirow{3}{*}{optimal Nash equilibrium} & {\sf NP}-hard &  {\sf NP}-hard \\ 
& Corollary~\ref{cor:max_NE-hard} &  \\ 
&  & \cite{CaiDaskalakis:2011:} \\
\hline 
\multirow{3}{*}{maximising potential} & tractable &  {\sf NP}-hard \\ 
& Corollary~\ref{cor:pot-pure-coord} & Corollary~\ref{cor:pot-anti-coord} \\ 
&  & \cite{Bramoulle:2007:GEB} and \cite{CaiDaskalakis:2011:} \\
\hline 
\end{tabular}
\end{center}
\label{table:classification}
\end{table}%



Observe that positive results exist for both welfare maximisation and potential maximisation in pure coordination polymatrix games.
While the complexity of computing the corresponding profiles for anti-coordination polymatrix games have already been resolved,\footnote{\cite{Bramoulle:2007:GEB} considers a binary action anti-coordination polymatrix game where preferences are homogeneous:\ every pairwise interaction is of the form $G_{ij}^s$ from Figure~\ref{fig:binaryGames}, with $\pay{i} = \pay{} = \frac{1}{2}$ for every player $i$.
\citeauthor{Bramoulle:2007:GEB} shows that identifying the potential maximiser and the welfare maximiser for this simple set up is equivalent to finding the maximum cut in a graph ({\sc Maximum Cut} is a well-known computationally hard problem).
Section 4 of \cite{CaiDaskalakis:2011:} shows that finding {\it any} Nash equilibrium (pure or mixed) is {\sf NP}-hard in the same environment explaining the result on optimal equilibria for anti-coordination games.} we will show that all four results can be obtained by invoking a result on a novel graph-theoretic problem on oriented graphs, hereafter orgraphs (these are directed graphs where if $ij$ is a directed edge then $ji$ is not).


The problem, introduced in \cite{DeligkasEiben:2023:arXiv}, is termed {\sc Maximum  Weighted Orgraph Partition} (MWOP).
{Informally, MWOP is described as follows (for the formal statement see Definition~\ref{def:MWOP}).}
Begin with an orgraph and assign a $2\times2$ matrix to every directed edge.\footnote{The convention we employ is that the vertex at the tail of a directed edge references the row of the matrix while the vertex at its head refers to the column of the matrix.}
Every vertex is assigned either $a$ or $b$, and the assignment of adjacent vertices determines the cell-entry of the matrix on the edge connecting them.
The objective of the MWOP problem is to partition the vertices into those assigned $a$ and those assigned $b$ such that the sum of the weights on each edge is maximised.

\cite{DeligkasEiben:2023:arXiv} completely characterise when the MWOP problem is tractable.
Surprisingly, the tractability depends only on the properties of the assigned matrices.
There are three families of matrices for which the MWOP problem is tractable.
The first two cases are immediate:\ when the maximum entry in every matrix is in the top left cell or in the bottom right cell, i.e., all vertices are assigned $a$ or all are assigned $b$.
The third tractable case is when, for every $2\times2$ matrix, the sum of the entries along the main diagonal is no smaller than the sum of the other two terms.
The proof of this case works by constructing a new auxiliary weighted graph and showing that the assignment that optimises the value is found by the minimum cut in the new graph; {\sc Minimum Cut} is a well-studied problem for which there are many efficient algorithms (see, e.g., \cite{korte2011combinatorial}).

To see how MWOP can be used to compute welfare-optimal and potential maximising outcomes in the coordination problems that we consider, the third tractable case of MWOP provides the key.
For any two player pure coordination game, both the welfare matrix and the potential matrix satisfy the property that the sum of terms along the main diagonal is at least as great as the sum of the other two terms.
For every two player anti-coordination game, the inequality goes the other way.
As an example, recall the games $G_{ij}^c$ and $G_{ij}^{s}$ from Figure~\ref{fig:binaryGames} and let us consider the issue of welfare.
With $i$'s strategy choice indexing the row and $j$'s indexing the column, the welfare matrices of both games are as given in Figure~\ref{fig:welfareMatrices}.
\begin{figure}[htp!]
\[
\left[\begin{array}{cc}
\pay{i}+\pay{j} & 0\\
0 & 2-\pay{i}-\pay{j}
\end{array}\right]
\hspace{.8in}
\left[\begin{array}{cc}
0 & \pay{i} + 1-\pay{j}\\
1-\pay{i} +\pay{j} & 0
\end{array}\right]
\]

\caption[]{Welfare matrices for $2\times2$ games $G_{ij}^c$ and $G_{ij}^{s}$ from Figure~\ref{fig:binaryGames}.}
\label{fig:welfareMatrices}
\end{figure}


The result follows almost immediately:\ welfare matrices for pure coordination polymatrix games are of the form of the left matrix in Figure~\ref{fig:welfareMatrices} and so maximising welfare is tractable.


Table~\ref{table:classification} indicates that it is not all good news for identifying outcomes in pure coordination polymatrix games.
In particular, computing the welfare optimal Nash equilibrium is intractable. 
We show this using the language game of \cite{Neary:2011:MGG,Neary:2012:GEB}.
The language game is a two-type threshold model.
Each player is a member of one of two groups, Group $A$ or Group $B$, and all players in the same group have the same threshold, either $\pay{A}$ or $\pay{B}$.
For this simple binary action environment, we will show show that efficient outcomes can be identified but that computing the welfare optimal Nash equilibrium is intractable.
\subsection{Computation of Nash Equilibria}\label{ssec:literatureComplexity}

While the theorem of \cite{Nash:1950:PNASUSA,Nash:1951:AM} ensures the existence of an equilibrium in every finite game, the proof is non-constructive and so provides no guide on how one might find one.
\cite{CDT} and \cite{DGP} showed that computing a Nash equilibrium is a problem that lies in the complexity class {\sf PPAD}.
Both papers explicitly use polymatrix games in their main constructions.\footnote{Some classical results concerning computing equilibria in polymatrix games existed already \citep{eaves1973polymatrix,howson1972equilibria,howson1974bayesian,miller1991copositive}.}


With the problem of computing of a Nash equilibrium resolved, attention has turned to the difficulty of finding outcomes in specific classes of games (as is our focus in this paper), to computing equilibria with certain properties (see for example both \cite{GilboaZemel:1989:GEB} and \cite{ConitzerSandholm:2008:GEB}), and also to that of computing approximate equilibria.
Recently, \cite{DFHM22-focs} proved a dichotomy for the complexity of finding {\it any} approximate well-supported Nash equilibrium in binary-action polymatrix games; in an approximate well-supported Nash equilibrium every player assigns positive probability only on actions that approximately maximize their payoff. 
\cite{DBLP:conf/innovations/BoodaghiansKM20}~studied the smoothed complexity on coordination polymatrix games, while~\cite{aloisio2021distance}~studied an extension of polymatrix games.
The complexity of constrained Nash equilibria for general polymatrix games were studied by~\cite{DeligkasFS17-constrained}, while~\citep{barman2015approximating,DeligkasFS20-tree-polymatrix,EGG-elkind2006nash,ortiz2017tractable} study games with tree-underlying structure.
Finally, \cite{DeligkasFIS16-experimental} provides an experimental comparison of various algorithms for polymatrix games.


\section{Coordination problems in networks}\label{sec:environment}

Section~\ref{ssec:polymatrix} begins with the definition of a binary action polymatrix game, and then introduces the two kinds of coordination problem that we consider:\ {\it pure coordination polymatrix games} and {\it anti-coordination polymatrix games}.
In Section~\ref{ssec:objects} we formally define the three kinds of strategies profiles that we are interested in.


\subsection{Binary action polymatrix games}\label{ssec:polymatrix}

A polymatrix game is the simplest kind of network game:\ a population of players reside on a graph where vertices represent players and edges represent pairwise connections; along every edge is a two player game;
every player's utility is the sum of payoffs earned with each neighbour where the same action must be used with each (so by definition each player has the same available strategies in all pairwise interactions in which they partake).


\begin{definition}\label{def:poly}
A binary action polymatrix game, $\Gcal$, consists of
\begin{itemize}

\item
an undirected graph $G = (V, E)$, where $V$ is a finite set $\set{1, . . . , n}$ of players, sometimes called vertices, and a finite set $E$ of edges, which are taken to be unordered pairs $\set{i,j}$ of players, $i \neq j$ (to economise on notation sometimes we will write $ij$ instead of $\set{i,j}$);

\item
for each player $i \in V$, a (common) set of strategies $S_i = S = \set{a, b}$;

\item
for each $\set{i,j}\in E$, a two-player game $(u^{ij}, u^{ji})$, where the players are $i$ and $j$, the strategy set for both player $i$ and player $j$ is $S$, and payoffs $u^{ij} : S_i \times S_j \to \Real$ and $u^{ji} : S_j\times S_i \to \Real$;

\item
for each player $i \in V$ and strategy profile $\stratBold{} = (\strat{1}, \dots, \strat{n}) \in \stratSpace = S^n$, the utility to player $i$ from $\stratBold{}$ is $U_{i}(\stratBold{}) = \sum_{\set{i,j} \in E} u^{ij}(\strat{i}, \strat{j})$.

\end{itemize}

\end{definition}

We are interested in two kinds of binary action polymatrix games that represent the two kinds of coordination problems. 
The first are pure coordination games and the second are anti-coordination games:\ in the former each player wants their neighbours to choose the same action, while in the latter every player wants the opposite.
Given that two player interactions are the building blocks of a polymatrix game, the formal definition of each kind of coordination problem that we employ is based on the structure of the equilibrium set of the two player games that make up the polymatrix game.

A strategy profile $\sbf{}^*$ is a pure strategy Nash equilibrium if no player can increase their utility by unilaterally changing their strategy.
Before defining this formally, note that we will occasionally write strategy profile $\sbf{}$ as $(\strat{i}, \sbf{-i})$ as this allows us to describe a strategy profile ``from the perspective'' of player $i$.


\begin{definition}
Strategy profile $\sbf{}^*$ is a {\it pure strategy Nash equilibrium}, if, for every player $i\in V$ and every $\strat{i} \in S$, it holds that $U_i(\strat{i}^{*},\sbf{-i}^*) \geq U_i(\strat{i},\sbf{-i}^*)$. 
A pure strategy Nash equilibrium, $\sbf{}^{*}$, is a {\it strict equilibrium} if the inequalities are strict.
\end{definition}

Given the above, we can now define the two kinds of coordination problem that we consider.

\begin{definition}\label{def:pure}
A binary-action polymatrix game is a {\it pure coordination polymatrix game} if, for every two player game, the action profiles $(a,a)$ and $(b,b)$ are strict equilibria.
\end{definition}

\begin{definition}\label{def:anti}
A binary-action polymatrix game is an {\it anti-coordination polymatrix game} if, for every two player game, the action profiles $(a,b)$ and $(b,a)$ are strict equilibria.
\end{definition}

Our definition of a pure coordination game, Definition \ref{def:pure}, is adapted from \cite{FosterYoung:1998:GEB}, where a two-person coordination game is defined as one in which both players have $k$ strategies and there exists a enumeration of each player's strategy set such that it is always a strict equilibrium for both players to play the strategy with the same index. 
We note that a two player anti-coordination game is a pure coordination game according to this definition (all one must do is swap the labelling of one of the player's strategy set).
However, this is not the case for anti-coordination polymatrix games in general (although it can be for specific network structures, e.g., a bipartite graph).

\subsection{Strategy profiles of interest}\label{ssec:objects}

We are interested in determining the following patterns of behaviour:\ a strategy profile that maximises welfare, a welfare optimal Nash equilibrium, and a profile that maximises potential.

\begin{definition}
Given a  binary-action polymatrix game, $\Gcal$, we write $\welf{}(\stratBold{})$ for the utilitarian social welfare evaluated at profile $\stratBold{}$. That is,
\begin{equation*}
\welf{}(\stratBold{}) := \sum_{i=1}^{n} U_{i}(\stratBold{}).
\end{equation*}
\end{definition}
The first profile that we are interested in is that which maximises social welfare, i.e., 
\begin{equation}\label{eq:welfMax}
\argmax_{\sbf{} \in \stratSpace} \welf{}(\sbf{}).
\end{equation}

We also care about pure strategy equilibria.
Let $\stratSpace^{*} \subseteq \stratSpace$ denote the set of pure strategy Nash equilibria.
The second object of interest is a pure strategy Nash equilibrium that maximises welfare, i.e.,
\begin{equation}\label{eq:bestNash}
\argmax_{\sbf{} \in \stratSpace^*} \welf{}(\sbf{}).
\end{equation}

%
%
%

While pure-coordination polymatrix games always possess two ``trivial'' pure Nash equilibria, it is easy to come up with an example where an anti-coordination polymatrix game does not possess a pure strategy equilibrium.
However, a well-studied class of game is {\it potential games} \citep{ShapleyMonderer:1996:GEB}.
Potential games possess many desirable properties.
One such property is that pure strategy equilibria correspond to local optima of the potential function so the existence of a pure strategy equilibrium is assured.


In a potential game, the incentive of each player to change their strategy can be expressed using a single function called the potential function.
\begin{definition}
A game is a {\it potential game} if there exists a function $\pot: \stratSpace \to \Real$ such that for every player $i$, for all $\sbf{-i}$, and all pairs of actions $\strat{i}', \strat{i}'' \in S_i$
\[
\pot(\strat{i}', \sbf{-i}) - \pot(\strat{i}'', \sbf{-i}) = U_i(\strat{i}', \sbf{-i}) - U_i(\strat{i}'', \sbf{-i}).
\]
\end{definition}


When every pairwise interaction of a polymatrix game is a potential game then the polymatrix game inherits this property, with the potential at any strategy profile equal to the sum of the potentials along every edge.
More formally, if $\pot^{ij} \colon S \times S \to \Real$ denotes the pairwise-potential function for the two-player game played between players $i$ and $j$, then the function $\pot \colon S^n \to \Real$ is a potential function for the polymatrix game, where
\begin{equation*}
\pot(\sbf{}) := \sum_{\set{i,j} \in E}\pot^{ij}(\strat{i}, \strat{j})
\end{equation*}

Whenever we consider potential maximisers, we will assume that every $2\times2$ game in a binary action polymatrix game is a potential game, so that we limit attention to {\it pairwise-potential binary action polymatrix games}.

The third and final pure strategy profile that we focus on is that which maximises the potential of a pairwise-potential polymatrix game, i.e.,
\begin{equation}\label{eq:potMax}
\argmax_{\sbf{} \in \stratSpace} \pot{}(\sbf{})
\end{equation}
where we note that since the profile that maximises the potential is guaranteed to be a pure strategy Nash equilibrium \citep{ShapleyMonderer:1996:GEB} and so optimising only over the set of pure strategies $\stratSpace$ is without loss of generality.

To sum up, given a coordination problem, the three patterns of behaviour that we are interested are those defined in \eqref{eq:welfMax}, \eqref{eq:bestNash}, and \eqref{eq:potMax}.




\section{Results}\label{sec:results}

In Section~\ref{ssec:welfare+potential} we show that the problem of maximising welfare and maximising the potential of a pure coordination polymatrix games are both tractable problems.
However, computing the two corresponding outcomes in anti-coordination polymatrix games is intractable.
While these two negative results were known already \citep{Bramoulle:2007:GEB, CaiDaskalakis:2011:}, we show how all four results can be recovered using a novel graph theoretic problem, the Maximum Weighted Orgraph Partition (MWOP) problem, introduced in \cite{DeligkasEiben:2023:arXiv}.
In Section~\ref{ssec:bestNash}, we prove that finding the welfare optimal Nash equilibrium in a pure coordination polymatrix game is intractable.
Computing the corresponding outcome for anti-coordination polymatrix games is known to be intractable since \cite{CaiDaskalakis:2011:} showed that computing any pure strategy Nash equilibrium in such games is \NP-hard.

\subsection{Maximising welfare and maximising potential}\label{ssec:welfare+potential}

Determining when the welfare maximising outcome and the potential maximising outcome is tractable follows by appealing to a novel graph-theoretic problem known as the Maximum Weighted Orgraph Partition (MWOP) problem introduced in \cite{DeligkasEiben:2023:arXiv}.


Throughout the paper we use standard terminology for both undirected graphs and directed graphs, see for example the textbook of \cite{Bang-JansenGutinBook}.
We define the main terms and any omitted can be found therein.

A {\it directed graph} (or just {\it digraph}) $D$ consists of a non-empty finite set $V(D)$ of elements called {\it vertices} and a finite set $A(D)$ of ordered pairs of distinct vertices called {\it arcs}. We call $V(D)$ the {\it vertex set} and $A(D)$ the {\it arc set} of $D$. 
An {\it oriented graph} is a directed graph without $2$-cycles (that is, if $ij \in A(D)$, then $ji \notin A(D)$).

A partition of the vertices in digraph $D$ is a mapping $P \colon V(D) \to \set{\one, \two}$.
Given a partition $P$, to keep notation simple, we write $P = (X_\one,X_\two)$ where $X_\one$ is the set of vertices assigned $\one$ and $X_\two$ those assigned $\two$.

\subsubsection{Maximum weighted orgraph partition}

In this section we state formally the MWOP problem and we characterise the conditions on matrices that render the problem tractable.
We begin by defining a family of matrices.

\begin{definition}
Let ${\cal F}$ denote a {\it family} of $2 \times 2$ matrices.
For each matrix $M \in {\cal F}$ let $m_{\one\one}$, $m_{\one\two}$, $m_{\two\one}$ and $m_{\two\two}$ denote the entries of $M$, so that 
\[
M = \left[\begin{array}{cc}
m_{\one\one} & m_{\one\two} \\
m_{\two\one} & m_{\two\two} \\
\end{array}\right].
\]
\end{definition}
 We say that a family $\Fcal$ of $2\times 2$ matrices is \emph{closed under multiplication by a positive scalar} if for all $2\times 2$ matrices $M$ and all $c\in \mathbb{R}^+$ it holds that $M\in \Fcal$ if and only if $c\cdot M\in \Fcal$.
We are now ready to define the MWOP problem.


\begin{definition}[\MWOP{}]\label{def:MWOP}
 Given a family of real $2 \times 2$ matrices $\Fcal$ that is closed under multiplication by a positive scalar, an \emph{instance of MWOP[$\Fcal$]} is given by the pair $(D, f)$, where
\begin{itemize}
    \item $D$ is an oriented graph on $n$ vertices;
    \item $f: A(D) \rightarrow \Fcal$, is an assignment of a matrix from the family of matrices $\Fcal$ to each arc in $D$.
\end{itemize}
Given a partition, $P=(X_\one,X_\two)$, of $V(D)$, the weight of an arc $ij \in A(D)$, denoted $\xi^P(ij)$, is defined as follows, where $M=\left[\begin{array}{cc}
m_{\one\one} & m_{\one\two} \\
m_{\two\one} & m_{\two\two} \\
\end{array} \right]= f(ij)$ is the matrix assigned to the arc $ij$.

\begin{equation}\label{eq:weights}
\xi^P(ij) = \left\{
\begin{array}{rcl}
m_{\one\one} & \mbox{if} & i,j \in X_\one \\
m_{\two\two} & \mbox{if} & i,j \in X_\two \\
m_{\one\two} & \mbox{if} & i \in X_\one \mbox { and } j \in X_\two \\
m_{\two\one} & \mbox{if} & i \in X_\two \mbox { and } i \in X_\one \\
\end{array}
\right.
\end{equation}
Given an orgraph $D$ and a partition $P$ of $V(D)$, the \emph{weight of $D$}, denoted by $\xi^P(D)$, is defined as the sum of the weights on every arc.
That is, $\xi^P(D) = \sum_{ij \in A(D)} \xi^P(ij)$. 
Finally, given an instance $(D,f)$ of MWOP[$\Fcal$],
the \emph{objective} of MWOP[$\Fcal$] is to find a partition $P$ of $V(D)$ that maximizes $\xi^P(D)$.
\end{definition}

{Note that \cite{DeligkasEiben:2023:arXiv} defined an instance of MWOP[$\Fcal$] in a way that they do not require $\Fcal$ to be closed under multiplication by a positive scalar but instead impose a scalar $c_{ij}\in\mathbb{R}^+$ on a each arc $ij$. However, these two definitions are equivalent, as we can always take a ``closure'' of the family by including all multiples of each matrix in $\Fcal$ and then for each arc in the digraph replaces the constant $c_{ij}$ and matrix $M_{ij}$ by the matrix $c_{ij}\cdot M_{ij}$.}


Before proceeding, let us pause and remark on the need for oriented graphs.
While polymatrix games do not specify directed connections, when assigning a matrix to the edge connecting a pair of players, imposing an orientation on that edge is more precise as it removes any perceived ambiguity.
Without doing so, statements about which vertex refers to the row of the matrix and which vertex refers to the column may cause confusion.

The main result in \cite{DeligkasEiben:2023:arXiv} is that the tractability of solving MWOP[$\Fcal$] depends on properties of the matrices in the family of matrices ${\cal F}$.
We now introduce three properties of matrices that will prove useful.

\begin{enumerate}[leftmargin=\parindent,align=left,label=\textbf{Property (\roman*):},ref=(\roman*)]
\item\label{property_i} $m_{\one\one} = \max\{ m_{\one\one}, m_{\two\two}, m_{\one\two}, m_{\two\one} \}$.
\item\label{property_ii} $m_{\two\two} = \max\{ m_{\one\one}, m_{\two\two}, m_{\one\two}, m_{\two\one} \}$.
\item\label{property_iii} $m_{\one\one} + m_{\two\two} \geq m_{\one\two} + m_{\two\one}$.
\end{enumerate}


The following is the main result in \cite{DeligkasEiben:2023:arXiv}.

\begin{thm}[\cite{DeligkasEiben:2023:arXiv}]\label{thm:main}
Given a family of $2\times 2$ matrices $\Fcal$ that is closed under multiplication by a positive scalar, MWOP[$\Fcal$] can be solved in polynomial time if either all matrices in ${\cal F}$ satisfy Property~\ref{property_i}, or all matrices in ${\cal F}$ satisfy Property~\ref{property_ii}, or all matrices in ${\cal F}$ satisfy Property~\ref{property_iii}.
Otherwise, MWOP[$\Fcal$] is {\sf NP}-hard. 
\end{thm}

The result above decomposes MWOP[$\Fcal$] into four subcases. 
Those in which every matrix satisfies one of the three properties and ``everything else''.
We emphasise that there are many subcases covered by ``everything else''.

The proof of cases covered by Property~\ref{property_i} and Property~\ref{property_ii} are almost immediate.
Clearly the optimal partition $(X_\one, X_\two)$ is either $\big(V(D), \emptyset\big)$ or $\big(\emptyset, V(D)\big)$.




The case where every matrix satisfies Property~\ref{property_iii} is directly applicable to binary action pure coordination polymatrix games since both the welfare matrix and the potential matrix of every two-player game possess this property.
We now sketch the proof from \cite{DeligkasEiben:2023:arXiv} for this case.
(Given its importance, we repeat the proof in \ref{app:proofMain}.)
Let us suppose there is a vertex $i$ in $D$, with two out-neighbours that for the purpose of this explainer we label as $j_1$ and $j_2$.
The out-neighbourhood of $i$ is depicted below in Figure~\ref{fig:neighbourhood}, with the matrices on arcs $ij_1$ and $ij_2$ denoted $M_{{1}}$ and $M_{{2}}$ respectively.
(Recall that the convention we employ is that the tail vertex of any arc indexes the row of the matrix on that arc.)
Both $j_1$ and $j_2$ may have other neighbours (either in- or out-neighbours) that are omitted from the image.
This possibility is conveyed by the three lines emitting from each.
For the sake of exposition, let us suppose that the off diagonal entries in both matrices $M_{{1}}$ and $M_{{2}}$ are zero.

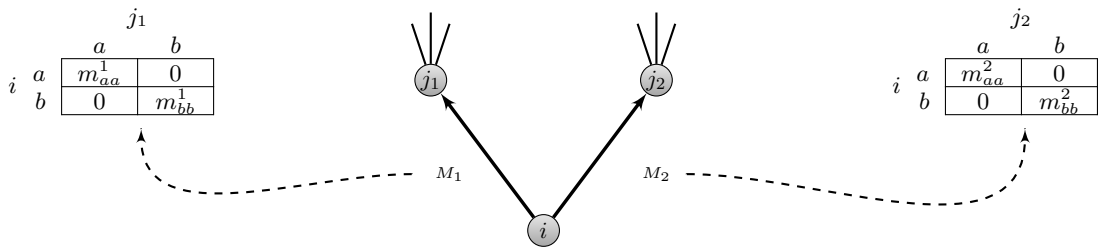
\begin{figure}[hbt!]
\centering
\tikzstyle{vertexW}=[circle,dotted, draw, top color=gray!1, bottom color=gray!1, minimum size=11pt, scale=0.02, inner sep=0.99pt]
\tikzstyle{vertexX}=[circle,draw, top color=gray!10, bottom color=gray!70, minimum size=10pt, scale=0.9, inner sep=0.5pt]
\tikzstyle{vertexY}=[circle,draw, top color=black!10, bottom color=gray!70, minimum size=15pt, scale=0.8, inner sep=0.4pt]
\tikzset{arc/.style = {->,> = latex'}}
\begin{tikzpicture}[scale=0.5]



\node (v1) at (1.0,4.0) [vertexY] {$j_{1}$}; 
\node (v2) at (7.0,4.0) [vertexY] {$j_{2}$};
\node (u) at (4.0,0.0) [vertexY] {$i$};

\draw (1.5,1.5) node {{\tiny $ M_{{1}}$}};
\draw (7,1.5) node {{\tiny $ M_{{2}}$}};

\node (c1) at (0.5,1.5) [vertexW] {$$}; 
\node (c11) at (-6.7,2.7) [vertexW] {$$};

\node (c2) at (7.8,1.5) [vertexW] {$$}; 
\node (c22) at (16.8,2.7) [vertexW] {$$};

\draw [arc, thick, dashed] (c1) to [out=180, in=270] (c11);
\draw [arc, thick, dashed] (c2) to [out=0, in=270] (c22);

\draw (-7.5,4.6) node {{\footnotesize 
\begin{game}{2}{2}[$i$][$j_1$]
& $a$ & $b$\\
$a$ &$m^{1}_{\one\one}$ & $0$\\
$b$ &$0$ &$m^{1}_{\two\two}$
\end{game}
}};

\draw (16.0,4.6) node {{\footnotesize 
\begin{game}{2}{2}[$i$][$j_2$]
& $a$ & $b$\\
$a$ &$m^{2}_{\one\one}$ & $0$\\
$b$ &$0$ &$m^{2}_{\two\two}$
\end{game}
}};

\draw [arc,line width=1.3pt] (u) -> (v1); 
\draw [arc,line width=1.3pt] (u) -> (v2); 

\draw [thick] (v1) -- (0.5,5.5);
\draw [thick] (v1) -- (1.0,5.8);
\draw [thick] (v1) -- (1.5,5.5);

\draw [thick] (v2) -- (6.5,5.5);
\draw [thick] (v2) -- (7.0,5.8);
\draw [thick] (v2) -- (7.5,5.5);


\end{tikzpicture}
\caption{The neighbourhood of vertex $i$ in $D$.}\label{fig:neighbourhood}
\end{figure}



The proof proceeds by constructing a new weighted undirected graph, $H$, with vertex set $V(H) = V(D) \cup \set{x, y}$ as follows.
Let $G(D)$ denote the undirected graph obtained from $D$ by ignoring orientations.
Let $E(H) = E(G(D)) \cup \set{xi, yi \; | \; i \in V(D)}$ and initially assign all edges zero weight.


In Figure~\ref{fig:ust} below, we display the neighbourhoods of $i$, $j_1$ and $j_2$ in $H$.
We have coloured vertex $j_1$ in \blue{blue} and $j_2$ in \red{red}.
There is a weight assigned to each edge that captures precisely the loss that occurs from assigning a vertex to one of the two partitions; the exact way the weight is computed can be found in the proof at the Appendix.

%
%
%
We emphasise that the weights currently assigned to some of the edges are incomplete as we include only the weight contributed due to vertex $i$ and its neighbours.\footnote{The weights are also colour coded to reflect the neighbour of $i$ that generated the weight.}
That is, since vertices $j_{1}$ and $j_{2}$ have more neighbours in $D$ other than $i$, there should be more weight on some of the edges.
We use the symbol $\Ocal$ to connote that there are be omitted terms;
the omitted weights may differ for each edge even though we use $\Ocal$ for each.

\begin{figure}[hbt!]
\centering
\tikzstyle{vertexX}=[circle,draw, top color=gray!10, bottom color=gray!70, minimum size=10pt, scale=0.9, inner sep=0.5pt]
\tikzstyle{vertexY}=[circle,draw, top color=black!10, bottom color=gray!70, minimum size=15pt, scale=0.8, inner sep=0.4pt]
\tikzstyle{vertexYb}=[circle,draw, top color=blue!40, bottom color=blue!40, minimum size=15pt, scale=0.8, inner sep=0.4pt]
\tikzstyle{vertexYr}=[circle,draw, top color=red!40, bottom color=red!40, minimum size=15pt, scale=0.8, inner sep=0.4pt]
\begin{tikzpicture}[scale=0.8]


\node (v1) at (6.0,3.0) [vertexYb] {$j_{1}$}; 
\node (v2) at (6.0,9.0) [vertexYr] {$j_{2}$};
\node (u) at (6.0,6.0) [vertexY] {$i$};

\node (s) at (0.0,6.0) [vertexY] {$x$}; 
\node (t) at (12.0,6.0) [vertexY] {$y$};

\draw [thick] (u) -- (v1); 
\draw [thick] (u) -- (v2);

\draw (6.8,7.8) node {{\tiny $\frac{\red{m^{2}_{\one\one}} + \red{m^{2}_{\two\two}}}{2}$}};
\draw (6.8,4.5) node {{\tiny $\frac{\blue{m^{1}_{\one\one}} + \blue{m^{1}_{\two\two}}}{2}$}};

\draw (9.7,8.0) node {{\tiny $\frac{-\red{m^{2}_{\one\one}}}{2} + \Ocal$}};
\draw (9.7,4.0) node {{\tiny $\frac{-\blue{m^{1}_{\one\one}}}{2} + \Ocal$}};

\draw (2.5,8.0) node {{\tiny $ -  \frac{\red{m^{2}_{\two\two}}}{2} + \Ocal$}};
\draw (2.5,4.0) node {{\tiny $ -  \frac{\blue{m^{1}_{\two\two}}}{2} + \Ocal$}};
\draw (3.5,6.4) node {{\tiny $ - \frac{\blue{m^{1}_{\two\two}}}{2} -  \frac{\red{m^{2}_{\two\two}}}{2}$}};
\draw (8.5,6.4) node {{\tiny $\frac{-\blue{m^{1}_{\one\one}}}{2} +  \frac{-\red{m^{2}_{\one\one}}}{2}$}};


\draw [thick] (s) -- (u); 
\draw [thick] (s) -- (v1); 
\draw [thick] (s) -- (v2); 

\draw [thick] (t) -- (u); 
\draw [thick] (t) -- (v1); 
\draw [thick] (t) -- (v2); 

\draw [thick] (v2) -- (6.0,10.5);
\draw [thick] (v2) -- (5.7,10.3);
\draw [thick] (v2) -- (6.3,10.3);

\draw [thick] (v1) -- (6.0,1.5);
\draw [thick] (v1) -- (5.7,1.7);
\draw [thick] (v1) -- (6.3,1.7);


\end{tikzpicture}
\caption{The neighbourhood of $i$ in $H$. Symbol $\Ocal$ conveys omitted weights.}\label{fig:ust}
\end{figure}
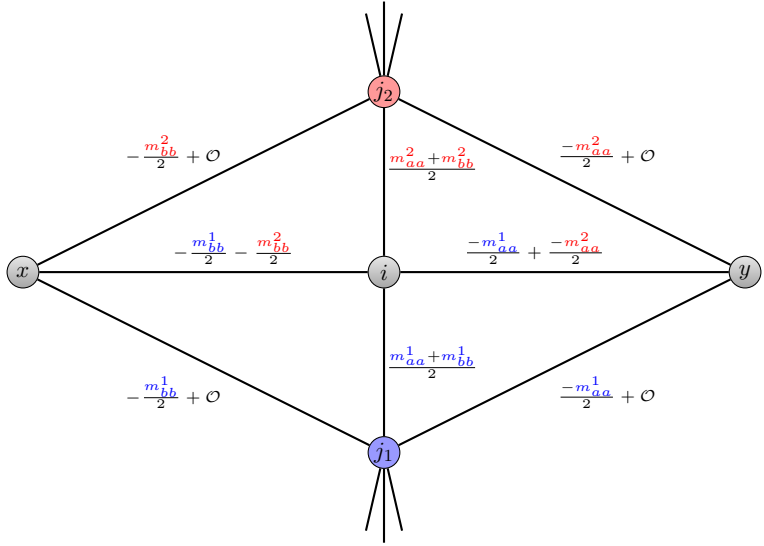

For every edge to which either $x$ is adjacent or $y$ is adjacent, the weights may be negative.
In fact, in the example of Figure~\ref{fig:ust} all are negative.
Now, to each edge adjacent to either $x$ or $y$ we add a weight of magnitude equal to that of the largest negative weight in the graph such that the weight on every edge becomes non-negative.
It is shown in the proof that while doing this will change the value of the solution, doing this will not alter the partition that maximises the sum of the weights.



With non-negative weight on each edge, it is straightforward to show that the partition that maximises the sum of the weights on every edge is equivalent to finding a minimum $x$-$y$ cut in the graph $H$.
One such cut is shown in Figure~\ref{fig:cut} below, where $i, j_1 \in X_\one$ and $j_2 \in X_\two$.
The edges remaining after the cut are those that contribute weight and are depicted in bold.

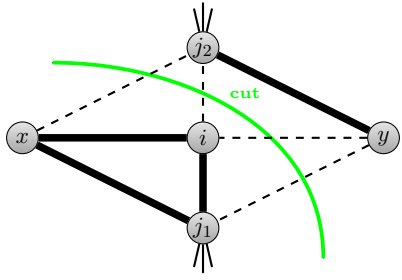
\begin{figure}[hbt!]
\centering
\tikzstyle{vertexX}=[circle,draw, top color=gray!10, bottom color=gray!70, minimum size=10pt, scale=0.9, inner sep=0.5pt]
\tikzstyle{vertexY}=[circle,draw, top color=black!10, bottom color=gray!70, minimum size=15pt, scale=0.8, inner sep=0.4pt]
\tikzstyle{vertexW}=[circle,dotted, draw, top color=gray!1, bottom color=gray!1, minimum size=11pt, scale=0.01, inner sep=0.99pt]
\begin{tikzpicture}[scale=0.4]


\node (v1) at (6.0,3.0) [vertexY] {$j_{1}$}; 
\node (v2) at (6.0,9.0) [vertexY] {$j_{2}$};
\node (u) at (6.0,6.0) [vertexY] {$i$};

\node (s) at (0.0,6.0) [vertexY] {$x$}; 
\node (t) at (12.0,6.0) [vertexY] {$y$};

\node (c1) at (1.0,8.5) [vertexW] {$$}; 
\node (c2) at (10.0,2.0) [vertexW] {$$};

\draw [line width=0.5mm, green] (c1) to [out=0, in=90] (c2); 

\draw [line width=1.0mm] (u) -- (v1); 
\draw [thick, dashed] (u) -- (v2);


\draw (7.4,7.5) node {{\tiny \green{\textbf{cut}}}};
%
%


\draw [line width=1.0mm] (s) -- (u); 
\draw [line width=1.0mm] (s) -- (v1); 
\draw [thick, dashed] (s) -- (v2); 

\draw [thick, dashed] (t) -- (u); 
\draw [thick, dashed] (t) -- (v1); 
\draw [line width=1.0mm] (t) -- (v2); 

\draw [thick] (v2) -- (6.0,10.5);
\draw [thick] (v2) -- (5.7,10.3);
\draw [thick] (v2) -- (6.3,10.3);

\draw [thick] (v1) -- (6.0,1.5);
\draw [thick] (v1) -- (5.7,1.7);
\draw [thick] (v1) -- (6.3,1.7);


\end{tikzpicture}
\caption{Remaining edges after an $x$-$y$ cut in which $i, j_{1} \in X_{\one}$ and $j_2 \in X_\two$.}\label{fig:cut}
\end{figure}

We conclude the sketch of the proof of part (iii) by noting that a minimum $x$-$y$ cut in a graph can be performed in polynomial time (see, e.g., \cite{korte2011combinatorial}).

When it comes to applying this result to polymatrix games, we will interpret vertex $x$ as a dummy player who is forced to choose action $\one$ and  vertex $y$ as a dummy player who must choose action $\two$.
Since every player must choose either action $\one$ or action $\two$, a player must be connected to either $x$ or $y$ but not both.
Hence the need for a cut.

For a proof of the hardness of every other possible case, we refer the reader to \cite{DeligkasEiben:2023:arXiv}.

In the next two subsections we show how to utilise the theorem of \cite{DeligkasEiben:2023:arXiv} to derive as corollaries several results, both positive and negative, for various classes of binary-action polymatrix games.
First we focus on maximising welfare and then on maximising potential.

\subsubsection{Maximising social welfare via MWOP}
Consider a binary action polymatrix game $\Gcal$ on graph $G$ and fix an arbitrary orientation of every edge to generate an oriented graph $D_G = \big(V(D_G), A(D_G)\big)$.
Now suppose that for every arc $ij \in A(D_G)$, that $i$ is the row player and that $j$ is the column player of the two-player game played over edge $\set{i,j} \in E(G)$.
We define the \emph{welfare matrix} on arc $ij$ as follows: 
$W^{ij}:=
\left[\begin{array}{cc}
w^{ij}_{\one\one} & w^{ij}_{\one\two} \\
w^{ij}_{\two\one} & w^{ij}_{\two\two} \\
\end{array}\right]$, where $w^{ij}_{s_i, s_j} = u^{ij}(s_i,s_j)+u^{ji}(s_j,s_i)$ for $s_i, s_j\in \set{\one,\two}$.
In words, given a strategy profile $\sbf{} = (s_1\ldots, s_n)$, the term $w^{ij}_{s_i, s_j}$, captures precisely the ``contribution'' of edge $\{i,j\}\in V(G)$ towards the utilitarian social welfare $\welf{}(\stratBold{})$.
Let $(D_G,f)$ be an instance of MWOP such that for all arcs $ij\in A(D_G)$ we have $f(ij)= W^{ij}$.


For any strategy profile $\sbf{}$, we use $P(\sbf{})$ to denote the partition of players into those choosing action $\one$ and those choosing $\two$.
The next lemma follows immediately.
\begin{lemma}
\label{lem:sw-to-mwdp}
For every possible strategy profile $\sbf{}$, it holds that $\welf(\sbf{}) = \xi^{P(\sbf{})}(D_G)$.
\end{lemma}

The above lemma gives us a straightforward correspondence between maximising the utilitarian social welfare of a polymatrix game and MWOP. In order to provide the complexity dichotomy for the problem of computing a strategy profile that maximises welfare, let us first formally define a parameterised version of the problem, which restricts the set of allowed welfare matrices in an instance of the problem.
More precisely, let $\Fcal$ be a family of $2\times 2$ matrices.
Then {\sc Welfare Maximisation}[$\Fcal$] is defined as follows.

\begin{definition}[\WelfareProblem]
Given a binary action polymatrix game $\Gcal$ on graph $G$ and an orientation $D_G$ of $G$ such that for every arc $ij\in A(D_G)$ the welfare matrix $W^{ij}$ is in $\Fcal$, find a strategy profile $\sbf{}$ that maximises the utilitarian social welfare $\welf(\sbf{})$.
\end{definition}

We highlight that the chosen orientation does not matter for the complexity of the problem.
This is if any matrix $M$ satisfies one of the Properties \ref{property_i}-\ref{property_iii}, then its transpose $M^T$ will satisfy the same property. 




Taken together, Lemma~\ref{lem:sw-to-mwdp} and the theorem of \cite{DeligkasEiben:2023:arXiv} yield the following clean complexity dichotomy for the \WelfareProblem~ problem.


\begin{theorem}
\label{thm:sw-general}
Given a collection of $2\times 2$ matrices $\Fcal$ that is closed under multiplication by a positive scalar, 
{\sc Welfare Maximisation}\emph{[$\Fcal$]} can be solved in polynomial time if either all matrices in ${\cal F}$ satisfy Property~\ref{property_i}, or all matrices in ${\cal F}$ satisfy Property~\ref{property_ii}, or all matrices in ${\cal F}$ satisfy Property~\ref{property_iii}.
Otherwise, {\sc Welfare Maximisation}\emph{[$\Fcal$]} is {\sf NP}-hard. 
\end{theorem}



{It can easily be verified that the welfare-matrix $W^{ij}$ of a $2\times2$ pure-coordination game between players $i$ and $j$ has $w^{ij}_{\one\one} + w^{ij}_{\two\two} \geq w^{ij}_{\one\two} + w^{ij}_{\two\one}$.
Similarly, for all $2\times2$ games of anti-coordination, such an inequality never holds.
Given this, Theorem~\ref{thm:sw-general} immediately yields the following two corollaries.


\begin{corollary}
\label{cor:pure-coord}
The problem of maximising social welfare in binary-action pure-coordination polymatrix games is solvable in polynomial time.
\end{corollary}

\begin{corollary}
\label{cor:anti-coord}
The problem of maximising social welfare in binary-action anti-coordination polymatrix games is {\sf NP}-hard. 
\end{corollary}


\subsubsection{Maximising potential via MWOP}

We now show how to use MWOP to maximise the potential in a binary action polymatrix game $\Gcal$ on graph $G$.
We assume that every two player game is a potential game thereby ensuring that the $\Gcal$ is too.
Again fix an orientation on every edge $\set{i,j} \in E(G)$ with $i$ indexing the row and $j$ the column.
The potential matrix on arc $ij$ is given by $\pot^{ij} := \left[\begin{array}{cc}
\phi^{ij}_{\one\one} & \phi^{ij}_{\one\two} \\
\phi^{ij}_{\two\one} & \phi^{ij}_{\two\two} \\
\end{array}\right]$.


Let $\pot(\sbf{})$ be the sum of the pairwise potentials along every arc.
If we denote by $P(\sbf{})$ the partition associated with strategy profile $\sbf{}$, we get the following lemma for potential that parallels Lemma~\ref{lem:sw-to-mwdp} for welfare.
\begin{lemma}
\label{lem:pot-to-mwdp}
For every strategy profile $\sbf{}$ of a binary action polymatrix game, it holds that $\pot(\sbf{}) = \xi^{P(\sbf{})}(D)$.
\end{lemma}

Before we state the main result for maximising the potential, let us first formally define the parameterised version of the problem depending on a family $\Fcal$ of allowed potential matrices.

\begin{definition}[\PotentialProblem] Given a binary action polymatrix game $\Gcal$ on graph $G$ and an orientation $D$ of $G$ such that for every arc $ij\in A(D)$, the game between $i$ and $j$ is a potential game with its potential matrix $\Phi^{ij}$ in $\Fcal$, find a strategy $\sbf{}$ that maximises the overall potential $\Phi(\sbf{})$ of $\Gcal$.    
\end{definition}

Lemma~\ref{lem:pot-to-mwdp} combined with the theorem of \cite{DeligkasEiben:2023:arXiv} yields the following.
To the best of our knowledge, this is the first example of a potential game where the identifying the potential maximiser is non trivial and yet doing so is possible.

\begin{theorem}
\label{thm:pot-general}
Given a collection of $2\times 2$ matrices $\Fcal$ that is closed under multiplication by a positive scalar, 
{\sc Potential Maximisation}\emph{[$\Fcal$]} can be solved in polynomial time 
if either all matrices in ${\cal F}$ satisfy Property~\ref{property_i}, or all matrices in ${\cal F}$ satisfy Property~\ref{property_ii}, or all matrices in ${\cal F}$ satisfy Property~\ref{property_iii}.
Otherwise, {\sc Potential Maximisation}\emph{[$\Fcal$]} is {\sf NP}-hard. 
\end{theorem}

Consider a two player binary action pure coordination game between players, labeled $i$ and $j$, that admits a potential, $\pot^{ij}$.
In this case it is immediate that $\phi^{ij}_{\one\one}+\phi^{ij}_{\two\two} \geq \phi^{ij}_{\one\two}+\phi^{ij}_{\two\one}$ (since both profiles $(\one, \one)$ and $(\two, \two)$ are pure strategy Nash equilibria and hence are local maximisers of the potential function $\pot^{ij}$).
For a two player game of anti-coordination the inequality holds strictly in the other direction since profiles $(\one, \two)$ and $(\two, \one)$ are equilibria.
Given this, Theorem~\ref{thm:pot-general} reveals another difference between pure-coordination and anti-coordination polymatrix games, when every two player game in each is a pairwise-potential game. 
In the latter class of games it is {\sf NP}-hard to maximise the potential function (this was already shown in \cite{Bramoulle:2007:GEB} and \cite{CaiDaskalakis:2011:}), whereas for the former class of games the problem can be solved in polynomial time.

\begin{corollary}
\label{cor:pot-pure-coord}
The problem of maximising the potential in binary action pure-coordination polymatrix games is solvable in polynomial time.
\end{corollary}

\begin{corollary}
\label{cor:pot-anti-coord}
The problem of maximising the potential in binary-action anti-coordination polymatrix games is {\sf NP}-hard. 
\end{corollary}

\subsection{Computing a welfare-optimal equilibrium}\label{ssec:bestNash}

In the previous section we showed that computing certain outcomes in pure-coordination polymatrix games is tractable but are not in anti-coordination games.
However, as we now show, when it comes to computing a ``best'' Nash equilibrium, things become almost immediately intractable.
We demonstrate this using arguably simplest pure-coordination polymatrix game with heterogeneous preferences that can played on arbitrary graphs: the language game of \cite{Neary:2011:MGG,Neary:2012:GEB}.

To introduce the language game we first introduce a special class of binary-action pure-coordination polymatrix games, termed {\em threshold games} \citep{NearyNewton:2017:JMID}.
A threshold game is a pure-coordination polymatrix game where every player $i$ is associated with a parameter $\gamma_i \in [0,1]$, and with each neighbour $j$, player $i$ plays the following game
\begin{figure}[htp!]
\hspace*{\fill}\begin{game}{2}{2}[${i}$][${j}$]
& $a$ & $b$\\
$a$ &$\pay{i},\pay{j}$ & $0,0$\\
$b$ &$0,0$ &1$-\pay{i}$, 1$-\pay{j}$
\end{game}\hspace*{\fill}
\end{figure}

Note that each player in the threshold model cares only about the number of neighbours choosing each action and not the identity of those neighbours.
That is, when player $i$ chooses action $a$ (respectively $b$) she earns $\pay{i}$ (respectively $1-\pay{i}$) from every neighbour who also choose $a$ (respectively $b$).

The 1-type threshold model, i.e., one wherein $\pay{i} = \pay{}$ is the same for all players $i$, is one of the most well-studied large population models.
When, for example, $\pay{} > \frac{1}{2}$, it has been shown that everyone playing strategy $\one$ is stochastically stable \citep{FosterYoung:1990:TPB, Young:1993:E} in a host of network structures \citep{KandoriMailath:1993:E,Ellison:1993:E,Peski:2010:JET,KreindlerYoung:2013:GEB}.
In related work, \cite{Morris:2000:RES} considers the issue of contagion in an infinite variant of the 1-type threshold model.
\citeauthor{Morris:2000:RES} supposes that the current state of play has all players choosing strategy $\two$.
He proposes {\it flipping} the behaviour of a finite of subset agents from $\two$ to $\one$ and then ``letting the system go'' via best-response dynamics.
The network property of {\it cohesiveness} is key for the superior outcome of everyone choose action $\one$ to eventually ``replace'' the inferior one of everyone playing $\two$.
Note that in all variations of the 1-type threshold model, what is viewed as the optimal outcome is not only agreed upon by one and all but is also trivial to identify. 
As such, questions of study include whether decentralised behaviour will ultimately lead there, if so how long will it take, and so on.



The language game of \cite{Neary:2011:MGG,Neary:2012:GEB} is a two-type threshold model.
The population is comprised of two groups, Group $A$ and Group $B$.
Every player is a member of one of two groups, and all players in the same group have the same threshold, either $\pay{A}$ or $\pay{B}$.
With two types of player, there are three types of pairwise interaction, $A$ with $A$, $B$ with $B$, and $A$ with $B$.
These games are shown below in Figure~\ref{fig:2typeModel}.
\begin{figure}[htb!]
\begin{center}
\[
\hspace{.6in} (u^{AA}, u^{AA}) \hspace{2.5in} (u^{BB}, u^{BB})
\]
\begin{game}{2}{2}[$A_{1}$][$A_{2}$]\label{game:GAA}
 	& $a$	& $b$\\
$a$	& $\pay{A}$, $\pay{A}$ & $0, 0$ \\
$b$ 	& $0, 0$ & $1 - \pay{A}$, $1 - \pay{A}$
\end{game}
\hspace{.3in}
\begin{game}{2}{2}[$B_{1}$][$B_{2}$]\label{game:GBB}
 	& $a$	& $b$\\
$a$	& $\pay{B}$, $\pay{B}$ & $0, 0$ \\
$b$ 	& $0, 0$ & $1 - \pay{B}$, $1 - \pay{B}$
\end{game}
\vspace{.1in}
\[
\hspace{.6in} (u^{AB}, u^{BA})
\]
\begin{game}{2}{2}[$A$][$B$]\label{game:GAB}
 	& $a$	& $b$\\
$a$	& $\pay{A}$, $\pay{B}$ & $0, 0$ \\
$b$ 	& $0, 0$ & $1 - \pay{A}$, $1 - \pay{B}$
\end{game}
\end{center}
\caption{Three pairwise interactions: $(u^{AA}, u^{AA})$, $(u^{BB}, u^{BB})$, and $(u^{AB}, u^{BA})$.}
\label{fig:2typeModel}
\end{figure}



Both within group interactions are symmetric.
The across group game is an asymmetric interaction whenever $\pay{A} \neq \pay{B}$.
When $\pay{B} < \frac{1}{2} < \pay{A}$ the across group game is a battle of the sexes and so a tension emerges:\ Group-$A$ players prefer the equilibrium where everyone coordinates on action $\one$ while Group-$B$ players prefer the opposite.
Even for this simple two-type threshold model, the following theorem shows that most cases are intractable.



\begin{theorem}
\label{thm:max_NE-hard}
The complexity of finding a welfare-optimal Nash equilibrium in the language game, a threshold-game with two thresholds, $\pay{A}, \pay{B}$, such that $0 \le \pay{B} \le \pay{A} \le 1$, is as follows.
\begin{description}
\item[1.] If $\pay{A} \leq 1/2$ the problem can be solved in polynomial time.
\item[2.] If $\pay{B} \geq 1/2$ the problem can be solved in polynomial time.
\item[3.] If $\pay{B} = 0$ and $\pay{A} = 1$ the problem can be solved in polynomial time.
\item[4.] In all other cases the problem is {\sf NP}-hard.
\end{description} 
\end{theorem}



The proof, that can be found in Appendix~\ref{app:proof_thm3}, proceeds by reducing the {\sc Minimum Transversal} problem in 3-uniform hypergraphs, a problem known to be {\sf NP}-hard \citep{GareyJohnson:1979:book}, to that of computing the welfare-optimal equilibrium in the language game.
Given Theorem~\ref{thm:max_NE-hard}, we can conclude the following:
\begin{corollary}
\label{cor:max_NE-hard}
It is {\sf NP}-hard to determine the welfare optimal Nash equilibrium in binary-action pure coordination polymatrix games.
\end{corollary}

\newpage

\newpage

\clearpage

\appendix

\section*{APPENDIX}\label{APP}

\section{Proofs}\label{app:proofs}

\subsection[Part (iii) of Theorem in]{Part (iii) of Theorem in \cite{DeligkasEiben:2023:arXiv}}\label{app:proofMain}

If $m_{11} + m_{22} \geq m_{12} + m_{21}$ for all matrices $M \in \Fcal$ then MWOP(${\Fcal}$) can be solved in polynomial time.

\begin{proof}
Let $(D,f)$ be an instance of MWOP(${\Fcal}$).
We will construct a new (undirected) multigraph $H$ with vertex set $V(D) \cup \{x,y\}$ as follows. Let $G(D)$ denote the undirected multigraph obtained from $D$ by removing
orientations of all arcs. 
Initially let $E(H) = E(G(D)) \cup \{xi,yj \; | \; i \in V(D) \}$ 
and let all edges in $H$ have weight zero. Now for each arc $ij \in A(D)$ we modify the weight function $w$ of $H$ as follows, where $M=f(uv)$ is the matrix associated
with the arc $ij$.

\begin{itemize}
\item Let $w(ij)=\frac{m_{11} + m_{22} - m_{12} - m_{21}}{2}$. Note $w(ij)$ is the weight of the undirected edge $ij$ in $H$ associated with the arc $ij$ in $D$ and that this weight is non-negative.
\item Add $\frac{-m_{22}}{2}$ to $w(xi)$.
\item Add $\frac{-m_{22}}{2}$ to $w(xj)$.
\item Add $\frac{m_{21}-m_{11}-m_{12}}{2}$ to $w(yi)$.
\item Add $\frac{m_{12}-m_{11}-m_{21}}{2}$ to $w(yj)$.
\end{itemize}

Let $\theta$ be the smallest possible weight of all edges in $H$ after we completed the above process ($\theta$ may be negative).
Now consider the weight function $w^*$ obtained from $w$ by 
subtracting $\theta$ from all edges incident with $s$ or $t$. That is, $w^*(ij)=w(ij)$ if $\{i,j\} \cap \{x,y\}=\emptyset$ and 
$w^*(ij)=w(ij) - \theta$ otherwise.  
Now we note that all $w^*$-weights in $H$ are non-negative.  We will show that for any $(x,y)$-cut, $(X_1,X_2)$ in $H$ (i.e., $(X_1,X_2)$ partitions 
$V(H)$ and $x \in X_1$ and $y \in X_2$) the $w^*$-weight of the cut is equal to $-w^P(D) - |V(D)|\theta$, where $P$ is the partition 
$(X_1 \setminus \{s\},X_2 \setminus \{y\})$ in $D$.  Therefore a minimum weight cut in $H$ maximizes $w^P(D)$ in $D$.

Let $(X_1,X_2)$ be any $(x,y)$-cut in $H$.  For every $u \in V(D)$ we note that exactly one of the edges $xi$ and $iy$ will belong to the cut.
Therefore, we note that the $w^*$-weight of the cut is $|V(D)|\theta$ less than the $w$-weight of the cut.
 It therefore suffices to show that the $w$-weight of the cut is $-w^P(D)$ (where  $P$ is the partition
$(X_1 \setminus \{x\},X_2 \setminus \{y\})$ of $V(D)$).

Let $ij \in A(D)$ be arbitrary and consider the following four possibilities.

\begin{itemize}
\item $i,j \in X_1$. In this case when we considered the arc $ij$ above (when defining the $w$-weights) we added  
$\frac{m_{21}-m_{11}-m_{12}}{2}$ to $w(yi)$ and $\frac{m_{12}-m_{11}-m_{21}}{2}$ to $w(yj)$.
So, we added the following amount to the $w$-weight of the  $(s,t)$-cut, $(X_1,X_2)$.
\[
\left( \frac{m_{21}-m_{11}-m_{12}}{2} + \frac{m_{12}-m_{11}-m_{21}}{2} \right) = - m_{11}
\]
\item $u,v \in X_2$. In this case when we considered the arc $ij$ above (when defining the $w$-weights) we added
$\left(\frac{-m_{22}}{2}\right)$ to $w(xi)$ and $\left(\frac{-m_{22}}{2}\right)$ to $w(xj)$.
So, we added the following amount to the $w$-weight of the  $(x,y)$-cut, $(X_1,X_2)$.
\[
\left( \frac{-m_{22}}{2} + \frac{-m_{22}}{2} \right) = -m_{22}
\]
\item $i \in X_1$ and $j \in X_2$. In this case when we considered the arc $ij$ above (when defining the $w$-weights) we 
let $w(ij)= \frac{m_{11} + m_{22} - m_{12} - m_{21}}{2}$ and added
$\left( \frac{-m_{22}}{2}\right)$ to $w(xj)$ and
$\frac{m_{21}-m_{11}-m_{12}}{2}$ to $w(yi)$.
So, we added the following amount to the $w$-weight of the  $(x,y)$-cut, $(X_1,X_2)$.
\begin{eqnarray*}
\left( \frac{m_{11} + m_{22} - m_{12} - m_{21}}{2} + \frac{-m_{22}}{2} + \frac{m_{21}-m_{11}-m_{12}}{2} \right) = 
m_{12} 
\end{eqnarray*}
\item $j \in X_1$ and $i \in X_2$. In this case when we considered the arc $ij$ above (when defining the $w$-weights) we
let $w(ij)= \frac{m_{11} + m_{22} - m_{12} - m_{21}}{2}$ and added
$ \left(\frac{-m_{22}}{2}\right)$ to $w(xi)$ and
$ \frac{m_{12}-m_{11}-m_{21}}{2}$ to $w(yj)$.
So, we added the following amount to the $w$-weight of the  $(x,y)$-cut, $(X_1,X_2)$.
\begin{eqnarray*}
\left( \frac{m_{11} + m_{22} - m_{12} - m_{21}}{2} + \frac{-m_{22}}{2} + \frac{m_{12}-m_{11}-m_{21}}{2} \right)  = 
m_{21}
\end{eqnarray*}
\end{itemize}

Therefore we note that in all cases we have added $-w^P(ij)$ to the $w$-weight of the  $(x,y)$-cut, $(X_1,X_2)$.
So the total $w$-weight of the  $(x,y)$-cut is $-w^P(D)$ as desired.

Analogously, if we have a partition $P=(X_1,X_2)$ of $V(D)$, then adding $x$ to $X_1$ and $y$ to $X_2$ we obtain a $(x,y)$-cut
with $w$-weight $-w^P(D)$ of $H$. 
As we can find a minimum $w^*$-weight cut in $H$ in polynomial time\footnote{Often the Min Cut problem is formulated for graphs rather than multigraphs (and $H$ is a multigraph), but we can easily reduce 
the Min Cut problem from multigraphs to graphs by merging all edges with the same end-points to one edge with the same end-points and with weight equal the sum of the weights of the merged edges.}
 we can find a  partition $P=(X_1,X_2)$ of $V(D)$ with 
minimum value of $-w^P(D)$, which corresponds to the maximum value of $w^P(G)$. Therefore, 
MWOP($\Fcal$) can be solved in polynomial time in this case.
\end{proof}

\subsection{Proof of Theorem~\ref{thm:max_NE-hard}}
\label{app:proof_thm3}

The proof of parts (i) and (ii) are straightforward. Both symmetric profiles are always Nash equilibria.
For part (i) all players choose $\two$ is the optimal equilibrium and for part (ii) all players choose $\one$. So we begin with the proof of part (iii).

\paragraph{Proof of part (iii):}

Let $\pay{A}=1$ and let $\pay{B}=0$.
Let $G$ be a graph and let $(A,B)$ be a partition of $V(G)$, so that if $i \in Y \Leftrightarrow \pay{i} = \pay{Y}$, where $Y \in \{ A,B\}$. 
We want to find a Nash equilibrium, $(X_{\one},X_{\two})$, of maximum welfare, where $X_s$ denotes the set of players that choose action $\strat{} \in \{\one, \two\}$. 

Let  $(X_{\one},X_{\two})$ be a Nash equilibrium of $G$. This implies that all vertices in $X_{\two} \cap A$ only have edges to vertices in $X_{\two}$ (as otherwise it is not a Nash equilibrium).
Note that this means that any connected component $C$ in $G[A]$ is either subset of $X_{\one}$ or a subset of $X_{\two}$, since otherwise one can find an edge from a vertex in $X_{\two} \cap A$ to a vertex in $X_{\one}$. 
Analogously, every connected component of $G[B]$ is either a subset of $X_{\one}$ or a subset of $X_{\two}$.
Finally, given a connected component $C_A$ of $G[A]$ and a connected component $C_B$ of $G[B]$, if there is an edge in $G$ between a vertex $i\in C_A$ and a vertex $j\in C_B$, then it is not possible that $C_B\subseteq X_{\one}$ and $C_A\subseteq X_{\two}$.

We now construct a digraph $D$ as follows.
Let $C_1^A, C_2^A, \ldots, C_l^A$ be the components of $G[A]$ and 
$C_1^B, C_2^B, \ldots, C_m^B$ the components of $G[B]$. Let $V(D)=\{x,y\} \cup \{a_1,a_2,\ldots,a_l\} \cup \{b_1,b_2,\ldots,b_m\}$.
For each $C_k^A$ add an arc from $x$ to $a_k$ of weight $2|E(C_k^A)|$ and for each $C_k^B$ add an arc from $b_k$ to $y$ of weight $2|E(C_k^B)|$.
Now for each $k \in \{1,2,\ldots,l\}$ and $\ell \in \{1,2, \ldots, m\}$ we add an arc from $a_k$ to $b_\ell$ of weight $|E(C_k^A,C_\ell^B)|$ and
if $|E(C_k^A,C_\ell^B)|>0$ then add an arc from $b_\ell$ to $a_k$ of infinite weight. This completes the construction of $D$.

Let $(X,Y)$ be a $(x,y)$-cut (i.e., $x \in X$ and $y \in Y$) in $D$ containing no arc of infinite weight.
Let the weight of $(X,Y)$ be $\xi(X,Y)$.
We will now construct a Nash equilibrium $(X_{\one},X_{\two})$ of $G$.
For every $a_k \in X$ add all vertices in $C_k^A$ to $X_{\one}$ and for every $a_k \in Y$ add all vertices in $C_k^A$ to $X_{\two}$.
Analogously, for every $b_\ell \in S$ add all vertices in $C_\ell^B$ to $X_{\one}$ and for every $b_\ell \in T$ add all vertices in $C_\ell^B$ to $X_{\two}$.
This defines a partition $(X_{\one},X_{\two})$.

As no arc of infinite weight exists in the cut, we note that that $(X_{\one},X_{\two})$ is a Nash equilibrium.
We will now show that the welfare of $(X_{\one},X_{\two})$ is $\welf{}(X_{\one},X_{\two})=2|E(G)|-|E(A,B)| - \welf{}(X,Y)$.
If $xa_k$ belongs to the cut then all arcs in $C_k^A$ contribute zero to welfare and zero to the formula for $\welf(X_{\one},X_{\two})$.
If $b_{\ell}y$ belongs to the cut then all arcs in $C_\ell^B$ contribute zero to the welfare and zero to the formula for $\welf{}(X_{\one},X_{\two})$.
If $a_kb_\ell$ belongs to the cut then all arcs in $E(C_k^A,C_\ell^B)$ contribute zero to the welfare and zero to the formula for $\welf{}(X_{\one},X_{\two})$.
All other arcs contribute two to the welfare if both end-points lie in $A$ or both end-points lie in $B$ and if exactly one
end-point lies in $A$ and the other in $B$ then the edge contributes one to the welfare.
So the welfare of $(X_{\one},X_{\two})$ is indeed $\welf{}(X_{\one},X_{\two})=2|E(G)|-|E(A,B)| - \welf{}(X,Y)$.

Conversely let $(Y_{\one},Y_{\two})$ be a Nash equilibrium of maximum weight. We saw above that all vertices in $C_i^A$ belong to $Y_{\one}$ or to $Y_{\two}$.
If they belong to $Y_{\one}$, then let $a_k$ belong to $X$ and otherwise let $a_k$ belong to $Y$. Analogously if all
vertices of $C_\ell^B$ belong to $Y_{\one}$ then let $b_\ell$ belong to $X$ and otherwise let $b_\ell$ belong to $Y$.
Furthermore let $x$ belong to $X$ and $y$ belong to $Y$, which gives us a $(x,y)$-cut $(X,Y)$.
We will show that the weight of the cut, $\xi(S,T)$, satisfies $\xi(Y_{\one},Y_{\two})=2|E(G)|-|E(A,B)| - \xi(X,Y)$, where $\xi(Y_{\one},Y_{\two})$ is the welfare of $(Y_{\one},Y_{\two})$.
As $(Y_{\one},Y_{\two})$ is exactly the Nash equilibrium we obtained above if we had started with $(X,Y)$, we note that
$\xi(Y_{\one},Y_{\two})=2|E(G)|-|E(A,B)| - \xi(X,Y)$ indeed holds.

Therefore, if the maximum welfare Nash equilibrium has welfare ${\rm wel}^{opt}$ and the minimum weight $(x,y)$-cut has weight ${\rm cut}^{opt}$ then ${\rm wel}^{opt}=2|E(G)|-|E(A,B)| - {\rm cut}^{opt}$.
As we can find a minimum $(x,y)$-cut in polynomial time (see, e.g., \cite{korte2011combinatorial}), this gives us a polynomial time algorithm for finding a welfare optimal Nash equilibrium when  $\pay{A}=1$ and $\pay{B}=0$.
So we are done.

\paragraph{Proof of part (iv):}

Recall that it suffices to prove the theorem for $0 < \pay{B} <1/2 < \pay{A} \le 1$.
Thus it is enough to prove Theorem \ref{mainNPwelfareNash} below.

We will start with two useful lemmas.

\begin{lemma}\label{lemmaAclique}
Let  $\pay{A}$ be any real such that $ 1/2 < \pay{A} \leq 1$.
Let $G$ be the complete graph with $\nA \ge 2$ vertices of type $A$ (and no vertices of type $B$).
The maximum welfare is obtained when all players choose action $\one$. Furthermore, if $(X_{\one},X_{\two}) \not= (V(G),\emptyset)$ then the welfare is at least 
$\min\{2\pay{A}(\nA-1), \nA(\nA-1)(2 \pay{A}-1)\}$ less than the optimum.
\end{lemma}

\begin{proof}
Define $G$ as in the statement.  Assume that $(X_{\one},X_{\two})$ is a partition of $V(G)$. 
The welfare of this partition is denoted by $\welf{}(X_{\one},X_{\two})$. First assume that  $(X_{\one},X_{\two}) \not= (V(G),\emptyset)$.
Then the following holds, where $q=|X_{\two}|$,

\[
\begin{array}{rcl}
\welf{}\big(V(G),\emptyset\big) - \welf{}(X_{\one},X_{\two}) & = & 
\binom{|X_{\two}|}{2}(2 \pay{A}) + |X_{\one}|\cdot |X_{\two}|(2\pay{A}) - \binom{|X_{\two}|}{2}\big((2 (1-\pay{A})\big) \\
& = & q(q-1)\big(\pay{A}-(1-\pay{A})\big) + 2(\nA-q)q\pay{A} \\
& = & q\left( (q-1)(2\pay{A}-1) + 2\nA\pay{A} - 2q\pay{A}  \right)   \\
& = & q\left( 2q\pay{A} -q  -2 \pay{A} + 1 + 2\nA\pay{A} - 2q\pay{A}  \right)   \\
& = & q\left( 2\pay{A}(\nA-1) -q + 1  \right)   \\
\end{array}
\]

Consider the function $f(q)=q\left( 2\pay{A}(\nA-1) -q + 1  \right)$ and note that $f'(q)=-2q + 2\pay{A}(a-1)+1$, which implies that 
the minimum value of $f(q)$ in any interval is found in one of its endpoints. As $1 \leq q \leq \nA$ the minimum value of $f(q)$ is therefore $f(1)$ or $f(\nA)$, which implies the following.
\[
\welf{}(V(G),\emptyset) - \welf{}(X_{\one},X_{\two})  \geq  \min\{f(1),f(\nA)\} = \min\{2\pay{A}(\nA-1), \nA(\nA-1)(2 \pay{A}-1)\}
\]
As the above minimum is positive, we note that if $(X_{\one},X_{\two}) \not= \big(V(G),\emptyset\big)$ then $(X_{\one},X_{\two})$ does not obtain the maximum welfare, so the maximum welfare is obtained when all players choose action $\one$.
Furthermore, if  $(X_{\one},X_{\two}) \not= \big(V(G),\emptyset\big)$ then the welfare is at least
$\min\{2\pay{A}(\nA-1), \nA(\nA-1)(2 \pay{A}-1)\}$ less than the optimum as seen above.
\end{proof}

\begin{corollary}\label{corBclique}
Let  $\pay{B}$ be any real such that $ 0 \leq \pay{B} < 1/2$.
Let $G$ be the complete graph with $\nB\ge 2$ vertices of type B (and no vertices of type A).
The maximum welfare is obtained when all players choose action $\two$ and if $(X_{\one},X_{\two}) \not= \big(\emptyset,V(G)\big)$ then the welfare is at least
$\min\{2(1-\pay{B})(\nB-1), \nB(\nB-1)(1-2\pay{B})\}$ less than the optimum.
\end{corollary}

\begin{proof}
    By replacing action $\two$ with $\one$ and vice versa and letting $\pay{B}^{new}=1-\pay{B}$ then we obtain an equivalent problem.
The following now completes the proof, by Lemma~\ref{lemmaAclique},
\[
 \min\{2\pay{B}^{new}(\nB-1), \nB\nB(2 \pay{B}^{new}-1)\} =  \min\{2(1-\pay{B})(\nB-1), \nB(\nB-1)(1-2\pay{B})\}. \qedhere
\]
\end{proof}

\begin{lemma} \label{NashEq}
Let  $\pay{A}$ and $\pay{B}$ be reals such that $0<\pay{B} \leq 1/2 \leq \pay{A} \leq 1$.
Let $G$ be any graph and let $(A,B)$ be a partition of $V(G)$.
The partition $(X_{\one},X_{\two})$ is a Nash equilibrium if and only if the following holds for all $i \in V(G)$.

\begin{itemize}
\item For every player $i \in X_{\one} \cap A$ we have $N(i) \cap X_{\two}=\emptyset$ or $\frac{|N(i) \cap X_{\one}|}{|N(i) \cap X_{\two}|} \geq \frac{1-\pay{A}}{\pay{A}}$.
\item For every player $i \in X_{\two} \cap A$ we have $N(i) \cap X_{\one}=\emptyset$ or $\frac{|N(i) \cap X_{\two}|}{|N(i) \cap X_{\one}|} \geq \frac{\pay{A}}{1-\pay{A}}$.  
\item For every player $i \in X_{\one} \cap B$ we have $N(i) \cap X_{\two}=\emptyset$ or $\frac{|N(i) \cap X_{\one}|}{|N(i) \cap X_{\two}|} \geq \frac{1-\pay{B}}{\pay{B}}$.
\item For every player $i \in X_{\two} \cap B$ we have $N(i) \cap X_{\one}=\emptyset$ or $\frac{|N(i) \cap X_{\two}|}{|N(i) \cap X_{\one}|} \geq \frac{\pay{B}}{1-\pay{B}}$.
\end{itemize}
\end{lemma}

\begin{proof}
  Let $(X_{\one},X_{\two})$ be a partition of $V(G)$.
  Let $i \in X_{\one} \cap A$ be arbitrary.
  If we move $i$ from $X_{\one}$ to $X_{\two}$ then $i$'s gain is $|N(i) \cap X_{\two}|(1-\pay{A}) - |N(i) \cap X_{\one}|\pay{A}$.
This gain is greater than zero if $\frac{|N(i) \cap X_{\one}|}{|N(i) \cap X_{\two}|} < \frac{1-\pay{A}}{\pay{A}}$, so $i$ satisfies the Nash equilibrium property 
if and only if $N(i) \cap X_{\two}=\emptyset$ or $\frac{|N(i) \cap X_{\one}|}{|N(i) \cap X_{\two}|} \geq \frac{1-\pay{A}}{\pay{A}}$.

  Now let $i \in X_{\two} \cap A$ be arbitrary.
  If we move $i$ from $X_{\two}$ to $X_{\one}$ then $u$'s gain is $|N(i) \cap X_{\one}|\pay{A} - |N(i) \cap X_{\two}|(1-\pay{A})$.
This gain is greater than zero if $\frac{|N(i) \cap X_{\two}|}{|N(i) \cap X_{\one}|} < \frac{\pay{A}}{1-\pay{A}}$, so $i$ satisfies the Nash equilibrium property 
if and only if $N(i) \cap X_{\one} = \emptyset$ or $\frac{|N(i) \cap X_{\two}|}{|N(i) \cap X_{\one}|} \geq \frac{\pay{A}}{1-\pay{A}}$.

%
 
 The statements concerning players in Group $B$ follow in an identical manner.
\end{proof}
 

\begin{theorem}\label{mainNPwelfareNash} Let  $\pay{A}$ and $\pay{B}$ be reals such that $0<\pay{B} < 1/2 < \pay{A} \leq 1$.
Then it is \NP-hard to find a maximum welfare Nash equilibrium in a threshold-game with two thresholds, $\pay{A}$~and~$\pay{B}$.
\end{theorem}

\begin{proof}
 We will show the \NP-hardness result via a reduction from \textsc{Minimum Transversal} problem in $3$-uniform hypergraphs, which is known to be \NP-hard~\citep{GareyJohnson:1979:book}. 
 A $3$-uniform hypergraph $H$ has a set of vertices $V(H)$ and a set of hyperedges $E(H)$, where every edge $e\in E(H)$ is a subset of vertices of size exactly $3$. 
 A transversal of $H$ is a set of vertices $X$ such that $X\cap e\neq \emptyset$ for every edge $e\in E(H)$.
 The problem \textsc{Minimum Transversal} asks to find a transversal of minimum size.

Let  $H$ be a $3$-uniform hypergraph where we want to find a minimum transversal.
Let  $\pay{A}$ and $\pay{B}$ be given reals, such that $0<\pay{B} < 1/2 < \pay{A} \leq 1$.
We first define the following constants.

\begin{enumerate}[leftmargin=\parindent,align=left,label=\textbf{ (\alph*):},ref=(\alph*)]
\item\label{constant_a} Let $\theta=6|E(H)| + 2|V(H)|\lceil 3|E(H)|\frac{1-\pay{B}}{\pay{B}(1-2\pay{B})}\rceil$. 
\item\label{constant_b} Let $x_B$ be the smallest possible positive integer such that $x_B(1-\pay{B})(\pay{A}-\pay{B})/\pay{B} > \theta$.
\item\label{constant_c}  Let $x_A$ be any integer such that the following holds (which is possible as $1/\pay{B} > 1$ (in fact, $1/\pay{B} > 2$).
\[
(x_B+3)\frac{1-\pay{B}}{\pay{B}} - \frac{1}{\pay{B}} < x_A < (x_B+3)\frac{1-\pay{B}}{\pay{B}}
\]
\item\label{constant_d}  Let $c_A$ be the smallest possible integer such that $c_A \geq x_A$ and $(c_A-1)(2\pay{A}-1) > \theta + 2|E(H)|\cdot(x_A+x_B)$.
\item\label{constant_e}  Let $c_B$ be the smallest possible integer such that $c_B \geq x_B$ and $(c_B-1)(1-2\pay{B}) > \theta + 2|E(H)|\cdot(x_A+x_B)$ and
$\frac{c_B-1}{|E(H)|} \geq  \frac{\pay{B}}{1-\pay{B}}$.
\item\label{constant_f} Let $z= \lceil 3|E(H)| \frac{1-\pay{B}}{\pay{B}(1-2\pay{B})} \rceil$. Note that $\theta=6|E(H)| + 2z|V(H)|$.
\end{enumerate}

We will now construct a graph $G$ and a partition $(A,B)$ of $V(G)$ such that a solution to our problem for $G$ will give us a minimum transversal in $H$. For an illustration, see Fig. \ref{fig:G}.

Let $C_A$ and $C_B$ be cliques, such that $|C_A|=c_A$ and $|C_B|=c_B$ (see \ref{constant_d} and \ref{constant_e}). 
For each edge $e \in E(H)$, let $r_e$ be a vertex and let $R=\cup_{e \in E(H)} \{r_e\}$.
For each vertex $u \in V(H)$, let $u'$ be a vertex and let $V'=\cup_{u \in V(H)} \{u'\}$.
Furthermore let $Z_u$ denote a set of $z$ vertices and let $Z=\cup_{u \in V(H)} Z_u$.
We will now let $V(G)=V(C_A) \cup V(C_B) \cup R \cup V' \cup Z$.

Let $E_1$ denote the edges in $C_A \cup C_B$. For every vertex $r_e \in R$ add $x_A$ edges to $C_A$ and add $x_B$ vertices to $C_B$.
Let $E_2$ denote these edges.
For every $e \in E(H)$ add an edge from $r_e$ to $u'$ if and only if $u \in V(e)$ in $H$. Let $E_3$ denote these edges.
For every $u \in V(H)$ add all edges between $u'$ and $Z_u$ to $G$. Let $E_4$ denote these edges.
Note that $|E_4|=|V(H)|\cdot z = |V(H)|\lceil 3|E(H)|\frac{1-\pay{B}}{\pay{B}(1-2\pay{B})}\rceil$ (by \ref{constant_f}).
This completes the description of $G$ and note that $E(G)=E_1 \cup E_2 \cup E_3 \cup E_4$.

Note that by \ref{constant_a}-\ref{constant_f} above all values of $\theta$, $x_B$, $x_A$, $c_A$, $c_B$ and $z$ are polynomial in $|V(H)|+|E(H)|$ (as $\pay{A}$ and $\pay{B}$ are considered 
constant). Therefore, $|V(G)|$ is also a polynomial in $|V(H)|+|E(H)|$. 
Let all vertices in $C_A$ belong to $A$ and all other vertices belong to $B$. That is $B = V(C_B) \cup R \cup V' \cup Z$.
This completes the definition of the partition $(A,B)$ of $V(G)$.

\begin{figure}[th]
  \begin{center}
\tikzstyle{vertexB}=[circle,draw, minimum size=10pt, scale=0.7, inner sep=0.9pt]
\tikzstyle{vertexA}=[rectangle,draw, minimum size=8pt, scale=0.9, inner sep=0.9pt]
\tikzstyle{vertexBs}=[circle,draw, minimum size=6pt, scale=0.5, inner sep=0.9pt]
\begin{tikzpicture}[scale=0.5]


  \draw (0,12) rectangle (8,14); \node at (4,13) {$C_A$};
  \node (a1) at (1,13.2) [vertexA]{};
  \node (a2) at (2,12.8) [vertexA]{};
  \node (a3) at (6,13.2) [vertexA]{};
  \node (a4) at (7,12.8) [vertexA]{};

  \draw (10,12) rectangle (18,14); \node at (14,13) {$C_B$};
  \node (a1) at (11,13.2) [vertexB]{};
  \node (a2) at (12,12.8) [vertexB]{};
  \node (a3) at (16,13.2) [vertexB]{};
  \node (a4) at (17,12.8) [vertexB]{};

  \node (r1) at (1,9) [vertexB]{$r_{e_1}$};
  \node (r2) at (5,9) [vertexB]{$r_{e_2}$};
  \node (r3) at (9,9) [vertexB]{$r_{e_3}$};
  \node at (13,9) {$\cdots$};
  \node (r4) at (17,9) [vertexB]{$r_{e_m}$};

  \draw[line width=0.03cm] (r1) to (2,12);
  \draw[line width=0.03cm] (r1) to (10.5,12);

  \draw[line width=0.03cm] (r2) to (3,12);
  \draw[line width=0.03cm] (r2) to (11.5,12);

  \draw[line width=0.03cm] (r3) to (4,12);
  \draw[line width=0.03cm] (r3) to (12.5,12);

  \draw[line width=0.03cm] (r4) to (6,12);
  \draw[line width=0.03cm] (r4) to (14.5,12);

  \node (u1) at (1,4) [vertexB]{$u_1'$};
  \node (u2) at (4,4) [vertexB]{$u_2'$};
  \node (u3) at (7,4) [vertexB]{$u_3'$};
  \node (u4) at (10,4) [vertexB]{$u_4'$};
  \node at (13.5,4) {$\cdots$};
  \node (u5) at (17,4) [vertexB]{$u_n'$};

  \draw[line width=0.03cm] (r1) to (u1);
  \draw[line width=0.03cm] (r1) to (u2);
  \draw[line width=0.03cm] (r1) to (u3);

  \draw[line width=0.03cm] (r2) to (u1);
  \draw[line width=0.03cm] (r2) to (u3);
  \draw[line width=0.03cm] (r2) to (u4);

  \draw[line width=0.03cm] (r3) to (8.6,8);
  \draw[line width=0.03cm] (r3) to (9,8);
  \draw[line width=0.03cm] (r3) to (9.4,8);

  \draw[line width=0.03cm] (r4) to (16.6,8);
  \draw[line width=0.03cm] (r4) to (17,8);
  \draw[line width=0.03cm] (r4) to (17.4,8);

  \draw (0,-1.5) rectangle (2,1); \node at (1,-0.7) {$Z_{u_1'}$};
  \node (z11) at (0.4,0.5) [vertexBs]{};
  \node (z12) at (1,0.5) [vertexBs]{};
  \node (z13) at (1.6,0.5) [vertexBs]{};
  \draw[line width=0.08cm] (u1) to (1,1);

  \draw (3,-1.5) rectangle (5,1); \node at (4,-0.7) {$Z_{u_2'}$};
  \node (z11) at (3.4,0.5) [vertexBs]{};
  \node (z12) at (4,0.5) [vertexBs]{};
  \node (z13) at (4.6,0.5) [vertexBs]{};
  \draw[line width=0.08cm] (u2) to (4,1);

  \draw (6,-1.5) rectangle (8,1); \node at (7,-0.7) {$Z_{u_3'}$};
  \node (z11) at (6.4,0.5) [vertexBs]{};
  \node (z12) at (7,0.5) [vertexBs]{};
  \node (z13) at (7.6,0.5) [vertexBs]{};
  \draw[line width=0.08cm] (u3) to (7,1);

  \draw (9,-1.5) rectangle (11,1); \node at (10,-0.7) {$Z_{u_4'}$};
  \node (z11) at (9.4,0.5) [vertexBs]{};
  \node (z12) at (10,0.5) [vertexBs]{};
  \node (z13) at (10.6,0.5) [vertexBs]{};
  \draw[line width=0.08cm] (u4) to (10,1);

  \draw (16,-1.5) rectangle (18,1); \node at (17,-0.7) {$Z_{u_n'}$};
  \node (z11) at (16.4,0.5) [vertexBs]{};
  \node (z12) at (17,0.5) [vertexBs]{};
  \node (z13) at (17.6,0.5) [vertexBs]{};
  \draw[line width=0.08cm] (u5) to (17,1);



 \end{tikzpicture}
\caption{The graph $G$ when $H$ is a $3$-uniform hypergraph with $m$ edges, including the edges $e_1=\{u_1,u_2,u_3\}$ and $e_2=\{u_1,u_3,u_4\}$.
The square vertices in $C_A$ denotes A-vertices and round vertices (everywhere else) denote B-vertices. Furthermore, the thick edges between
$u_i$ and $Z_{u_i}$ denotes that all edges are present.} \label{fig:G}
\end{center} \end{figure}
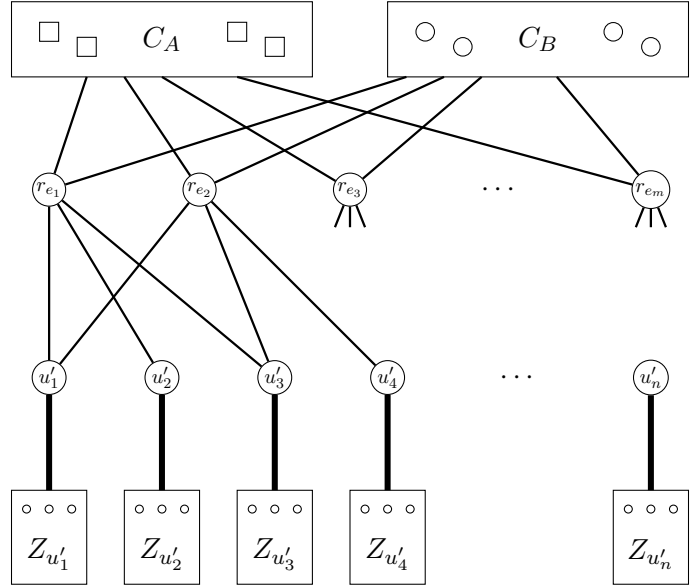

We will show that a welfare optimizing Nash equilibrium for $G$ gives us a minimum transversal for $H$, thereby proving the desired NP-hardness result. 
Assume we have a partition $(X_{\one},X_{\two})$ of $V(G)$ (where
all players in $X_{\one}$ choose action $\one$ and all players in $X_{\two}$ choose action~$\two$).

For every edge $ij \in G$ let the weight $\xi(ij)$ of $ij$ be the increase in welfare achieved because of that edge. That is, 
the following holds.
\[
\xi(ij) = \left\{ 
\begin{array}{rclcl}
2\pay{A} & & \mbox{if $i,j \in X_{\one}$ and $i,j \in A$} \\
2\pay{B} & & \mbox{if $i,j \in X_{\one}$ and $i,j \in B$} \\
\pay{A} + \pay{B}  & & \mbox{if $i,j \in X_{\one}$ and $|\{i,j\} \cap A|=1$} \\
2(1-\pay{A}) & & \mbox{if $i,j \in X_{\two}$ and $i,j \in A$} \\
2(1-\pay{B}) & & \mbox{if $i,j \in X_{\two}$ and $i,j \in B$} \\
2-\pay{A} - \pay{B}  & & \mbox{if $i,j \in X_{\two}$ and $|\{i,j\} \cap A|=1$} \\
0 & \hspace{0.3cm} & otherwise \\
\end{array}
\right.
\]
Define $W^*$ be the the value of the optimal welfare of $G-E_3-E_4$.
We will now prove the following claims.


\setcounter{equation}{0}
\renewcommand{\theequation}{\roman{equation}}

\begin{claim}\label{claim:A}
The optimal welfare of $G-E_3-E_4$ (with value $W^*$) is obtained for the partition $(X_{\one},X_{\two})$ if, and only if, the
following holds.
\begin{equation}
    V(C_A) \cup R \subseteq X_{\one} \text{ and } V(C_B) \subseteq X_{\two} \label{eq:A}
\end{equation}
\end{claim}

Furthermore any partition that does not satisfy \eqref{eq:A} above will have welfare at most $W^* - \theta$ (see \ref{constant_a}).

\renewcommand\qedsymbol{$\blacksquare$}

\begin{proof}[Proof of Claim~\ref{claim:A}]

Let $C^+$ be the component of $G-E_3-E_4$ containing the vertices $V(C_A) \cup V(C_B) \cup R$.
Let $W^+$ be the value of the welfare for $C^+$ for the partition $(V(C_A) \cup R,V(C_B))$ and let $(X_{\one},X_{\two})$
be any partition  $G-E_3-E_4$.
Let $W_A^+$ be the optimal welfare for $C_A$ and let  $W_B^+$ be the optimal welfare for $C_B$. 

Note that $\min\{2\pay{A}(c_A-1), c_A(c_A-1)(2 \pay{A}-1)\} \geq (c_A-1)(2 \pay{A}-1)$ as $2\pay{A}>1 \geq 2 \pay{A}-1$.
By Lemma~\ref{lemmaAclique} we note that if $V(C_A) \not\subseteq X_{\one}$ then the welfare of $C_A$ is at most
the following (by \ref{constant_d}):
\[
\begin{array}{rcl}
W_A^+ - \min\{2\pay{A}(c_A-1), c_A(c_A-1)(2 \pay{A}-1)\} & \leq &  W_A^+ - (c_A-1)(2 \pay{A}-1) \\
& < & W_A^+ - (\theta + 2|E(H)|\cdot(x_A+x_B)) \\
\end{array}
\]

Analogously, by Corollary~\ref{corBclique} and \ref{constant_e} we note that if  $V(C_B) \not\subseteq X_{\two}$ then the welfare of $C_B$ is at most
the following:
\[
\begin{array}{rcl}
W_B^+ - \min\{2(1-\pay{B})(c_B-1), c_B(c_B-1)(1-2\pay{B})\} & \leq &  W_B^+ - (c_B-1)(1-2 \pay{B}) \\
& < &  W_B^+ - (\theta + 2|E(H)|\cdot(x_A+x_B)) \\
\end{array}
\]

As the maximum welfare we can obtain from the edges in $E_2$ is not more than $2|E(H)|\cdot(x_A+x_B)$ we note that if $V(C_A) \not\subseteq X_{\one}$
or $V(C_B) \not\subseteq X_{\two}$ then the welfare of $(X_{\one},X_{\two})$ is at least $\theta$ less than $W_A^+ + W_B^+$ and therefore also 
at least $\theta$ less than $W^*$. So we may assume that $V(C_A) \subseteq X_{\one}$ and $V(C_B) \subseteq X_{\two}$.

Let $r_e \in R$ be arbitrary and assume that $r_e \in X_{\two}$. By moving $r_e$ to $X_{\one}$ we would increase the welfare of $G-E_3-E_4$ by 
$x_A (\pay{A} + \pay{B}) - x_B (2(1-\pay{B}))$. By \ref{constant_c} we note that $(x_B+3)\frac{1-\pay{B}}{\pay{B}} - \frac{1}{\pay{B}} < x_A$,
which implies the following (by \ref{constant_b}):

\[
\begin{array}{rcl}
x_A (\pay{A} + \pay{B}) - x_B (2(1-\pay{B}))  &\geq& \left( (x_B+3)\frac{1-\pay{B}}{\pay{B}} - \frac{1}{\pay{B}} \right) (\pay{A} + \pay{B}) - 2x_B (1-\pay{B})) \\
& = & x_B \left( \frac{1-\pay{B}}{\pay{B}}(\pay{A} + \pay{B}) - 2(1-\pay{B})) \right)    \\ 
& & + \left( 3 \frac{1-\pay{B}}{\pay{B}}  - \frac{1}{\pay{B}} \right) (\pay{A} + \pay{B}) \\
& = & x_B (1- \pay{B})\left( \frac{\pay{A}+\pay{B}}{\pay{B}} - 2 \right)    + \left( \frac{2-3\pay{B}}{\pay{B}} \right) (\pay{A} + \pay{B}) \\
& \geq & x_B (1- \pay{B})\left( \frac{\pay{A}-\pay{B}}{\pay{B}} \right)  \\
& > & \theta \hfill \qedhere \\
\end{array}
\]
\end{proof}

\begin{claim}\label{claim:B}
Let $(X_{\one},X_{\two})$ be any Nash equilibrium of $G$ that satisfies \eqref{eq:A} in Claim~\ref{claim:A}, then \eqref{eq:B} below holds.

\begin{equation}
\text{For every }e \in E(H)\text{ the vertex }r_e\text{ has a neighbour in }V' \cap X_{\one}. \label{eq:B}
\end{equation}
\end{claim}

\begin{proof}[Proof of Claim~\ref{claim:B}] Recall that by (A) in Claim~A we have  $V(C_A) \cup R \subseteq X_{\one}$ and $V(C_B) \subseteq X_{\two}$.
For the sake of contradiction assume that there exists a $r_e \in R$ such that all three neighbours of $r_e$ in $V'$ belong to $X_{\two}$. 
We will now show that $r_e$ does not satisfy the Nash equilibrium property.

The vertex $r_e$ currently has $x_A$ neighbours in $X_{\one}$ and $x_B+3$ neighbours in $X_{\two}$. So its current welfare is
$x_A \pay{B}$, but if we move $r_e$ to $X_{\two}$ then its welfare will become $(x_B+3)(1-\pay{B})$. By \ref{constant_c}, the following holds.

\[
(x_B+3)(1-\pay{B}) - x_A \pay{B}  >  (x_B+3)(1-\pay{B}) - \left((x_B+3)\frac{1-\pay{B}}{\pay{B}}\right) \pay{B} = 0 
\]

So, $r_e$ does not satisfy the Nash equilibrium property, a contradiction.
\end{proof}

\begin{claim}\label{claim:C}
Let $(X_{\one},X_{\two})$ be any partition of $V(G)$ that satisfies \eqref{eq:A} in Claim~\ref{claim:A}, \eqref{eq:B} in Claim~\ref{claim:B}, and \eqref{eq:C} below.
Then $(X_{\one},X_{\two})$ is a Nash equilibrium.
\begin{equation}
\text{For every } u \in V(H), \text{ either } \{u'\} \cup Z_u \subseteq X_{\one} \text{ or } \{u'\} \cup Z_u \subseteq X_{\two}. \label{eq:C}
\end{equation}
\end{claim}

\begin{proof}[Proof of Claim~\ref{claim:C}]
As vertices in $Z$ are only adjacent to vertices in the same partite set in $(X_{\one},X_{\two})$ as themselves they
do satisfy the Nash equilibrium condition. By Lemma~\ref{NashEq} we therefore need to show that the following holds for all $v_A \in V(C_A)$ and all $v_B \in V(C_B)$ and
all $r_e \in R$ and all $u' \in V'$ and all $z \in Z$.

\begin{itemize}
\item $\frac{|N(v_A) \cap X_{\one}|}{|N(v_A) \cap X_{\two}|} \geq \frac{1-\pay{A}}{\pay{A}}$ or $N(v_A) \cap X_{\two}=\emptyset$ (as $v_A \in A \cap X_{\one}$).
\item $\frac{|N(v_B) \cap X_{\two}|}{|N(v_B) \cap X_{\one}|} \geq \frac{\pay{B}}{1-\pay{B}}$ or $N(v_B) \cap X_{\one}=\emptyset$ (as $v_B \in B \cap X_{\two}$).
\item $\frac{|N(r_e) \cap X_{\one}|}{|N(r_e) \cap X_{\two}|} \geq \frac{1-\pay{B}}{\pay{B}}$ or $N(r_e) \cap X_{\two}=\emptyset$ (as $r_e \in B \cap X_{\one}$).
\item If $u' \in X_{\one}$, then $\frac{|N(u') \cap X_{\one}|}{|N(u') \cap X_{\two}|} \geq \frac{1-\pay{B}}{\pay{B}}$ or $N(u') \cap X_{\two}=\emptyset$ (as $u' \in B \cap X_{\one}$).
\item If $u' \in X_{\two}$, then $\frac{|N(u') \cap X_{\two}|}{|N(u') \cap X_{\one}|} \geq \frac{\pay{B}}{1-\pay{B}}$ or $N(u') \cap X_{\one}=\emptyset$ (as $u' \in B \cap X_{\two}$).
\end{itemize}

The first inquality holds as $N(v_A) \cap X_{\two}=\emptyset$. The second inequality holds due to the following (see \ref{constant_e}):
\[
\frac{|N(v_B) \cap X_{\two}|}{|N(v_B) \cap X_{\one}|} \geq  \frac{|V(C_B)|-1}{|E(H)|} \geq  \frac{\pay{B}}{1-\pay{B}}
\]
The third inequality holds because of the following (see \ref{constant_c}):
\[
\begin{array}{rcl}
\frac{|N(r_e) \cap X_{\one}|}{|N(r_e) \cap X_{\two}|} 
& \geq & \frac{x_A+1}{x_B+2} \\
& > & \frac{(x_B+3)\frac{1-\pay{B}}{\pay{B}} - \frac{1}{\pay{B}} +1}{x_B+2} \\
& = & \frac{(x_B+3)(1-\pay{B}) - 1 + \pay{B}}{\pay{B}(x_B+2)} \\
& = & \frac{1-\pay{B}}{\pay{B}} \\
\end{array}
\]

If $u' \in V' \cap X_{\one}$ then all neighbours of $u'$ belong to $X_{\one}$, so $N(u') \cap X_{\two}=\emptyset$, which satisfies the fourth equation.
Let $u' \in V' \cap X_{\two}$ be arbitrary.
By \ref{constant_f} we note that the following holds (as $\pay{B}<1/2$).
\[
z = \lceil 3|E(H)| \frac{1-\pay{B}}{\pay{B}(1-2\pay{B})} \rceil > |E(H)| \frac{\pay{B}}{1-\pay{B}}+1 
\]
 The fifth inequality now holds because of the following:
\[
\frac{|N(u') \cap X_{\two}|}{|N(u') \cap X_{\one}|} \geq \frac{z-1}{|E(H)|} > \frac{\pay{B}}{1-\pay{B}}. \hfill \qedhere
\]
\end{proof}

Let $T$ be any transversal of $H$.
Define a partition $(X_{\one},X_{\two})$ of $V(G)$ as follows. 
For every $u \in T$ let $V(Z_u) \cup \{u'\}$ belong to $X_{\one}$ and for 
every $u \in V(H) \setminus T$ let $V(Z_u) \cup \{u'\}$ belong to $X_{\two}$.
Then add $V(C_A) \cup R$ to $X_{\one}$ and add $V(C_B)$ to $X_{\two}$. 
We call this the {\em $G$-extension} of $T$.

\begin{claim}\label{claim:D}
The $G$-extension of $T$ is a Nash equilibrium with the following welfare, where $\epsilon_T$ is some value satisfying $0 \leq \epsilon_T  < 2z (1 - 2\pay{B})$.
\[
W^* + \epsilon_T +|V(H)|\cdot z (2(1-\pay{B})) - |T| \cdot  2z (1 - 2\pay{B})
\]
\end{claim}

\begin{proof}[Proof of Claim~\ref{claim:D}]
Let $(X_{\one},X_{\two})$ be the $G$-extension of a transversal $T$ in $H$.
As conditions (A), (B) and (C) are all satisfied in Claims~A, B and C, respectively we note that by Claim~C 
$(X_{\one},X_{\two})$ is a Nash equilibrium.

By Claim~A we note that $E_1 \cup E_2$ contribute $W^*$ towards the welfare of $(X_{\one},X_{\two})$.
Let $\epsilon_T$ be the welfare that the edges $E_3$ contribute and note that $0 \leq \epsilon_T \leq 2|E_3| = 6|E(H)|$. 
By \ref{constant_f}, $z= \lceil 3|E(H)| \frac{1-\pay{B}}{\pay{B}(1-2\pay{B})} \rceil > \frac{3|E(H)|}{1-2\pay{B}} $, which implies that
$0 \leq \epsilon_T \leq 6|E(H)| < 2z (1 - 2\pay{B})$. 

We will now compute the contribution to the welfare from the edges in $E_4$. Let $u' \in V'$ be arbitrary.
If $u' \in X_{\one}$ then the edge between $u'$ and $Z_u$ contribute $z(2\pay{B})$ to the welfare.
If $u' \in X_{\two}$ then the edge between $u'$ and $Z_u$ contribute $z(2(1-\pay{B}))$ to the welfare.
As there are $|T|$ vertices in $X_{\one} \cap V'$ we therefore get the following welfare for $(X_{\one},X_{\two})$.
\begin{align*}
\welf{(X_{\one},X_{\two})} &= W^* + \epsilon_T +(|V(H)|-|T|)\cdot z (2(1-\pay{B})) + |T| \cdot  2z\pay{B}\\
&= W^* + \epsilon_T +|V(H)|\cdot z (2(1-\pay{B})) - |T| \cdot  2z(1-2\pay{B})\hfill\qedhere
\end{align*}
\end{proof}


\begin{claim}\label{claim:E}
If we can find the welfare optimum Nash equilibrium of $G$ then we can find the minimum transversal in $H$. This implies that 
finding the welfare optimum Nash equilibrium for $\pay{A}$ and $\pay{B}$ is \NP-hard.
\end{claim}

\begin{proof}[Proof of Claim~\ref{claim:E}] Assume that $(X_{\one}^{opt},X_{\two}^{opt})$ is a welfare optimal Nash equilibrium of $G$ and let
$w(X_{\one}^{opt},X_{\two}^{opt})$ be the welfare of $(X_{\one}^{opt},X_{\two}^{opt})$. Let $T^{opt}$ be a minimum transversal of $H$.
We will show that the following holds, which completes the proof of Claim~\ref{claim:E}.
\[
|T^{opt}| = \left\lceil \frac{ W^* + |V(H)|\cdot z (2-2\pay{B})) - w(X_{\one}^{opt},X_{\two}^{opt}) }{2z(1-2\pay{B})} \right\rceil
\]
Let $(X_{\one},X_{\two})$ be the $G$-extension of $T^{opt}$, and note that, by Claim~\ref{claim:D}, $(X_{\one},X_{\two})$ is a Nash equilibrium satisfying the following,
where $0 \leq \epsilon_{T^{opt}} < 2z(1-2\pay{B})$.
\[
W(X_{\one},X_{\two}) = W^* + \epsilon_{T^{opt}} +|V(H)|\cdot z (2-2\pay{B})) - |T^{opt}| \cdot  2z(1-2\pay{B})
\]
This implies the following as $\epsilon_{T^{opt}} < 2z(1-2\pay{B})$ and $w(X_{\one}^{opt},X_{\two}^{opt}) \geq w(X_{\one},X_{\two})$.
\begin{align}
|T^{opt}| & =  \frac{ W^* + \epsilon_{T^{opt}} + |V(H)|\cdot z (2-2\pay{B})) - w(X_{\one},X_{\two}) }{2z(1-2\pay{B})} \nonumber \\
& =   \left\lceil \frac{ W^* + |V(H)|\cdot z (2-2\pay{B})) - w(X_{\one},X_{\two}) }{2z(1-2\pay{B})} \right\rceil \label{eq:asterix}\\
& \geq   \left\lceil \frac{ W^* + |V(H)|\cdot z (2-2\pay{B})) - w(X_{\one}^{opt},X_{\two}^{opt}) }{2z(1-2\pay{B})} \right\rceil \nonumber
\end{align}

By Claim~\ref{claim:D}, $w(X_{\one}^{opt},X_{\two}^{opt}) \geq w(X_{\one},X_{\two}) > W^*$. 
If conditions (A) in Claim~A is not satisfied for $(X_{\one}^{opt},X_{\two}^{opt})$ then by Claim~A the welfare of the edges $E_1 \cup E_2$ is at most $W^* - \theta$. 
This implies the following (by \ref{constant_f}).
\[
\begin{array}{rcl}
w(X_{\one}^{opt},X_{\two}^{opt}) & \leq & W^* - \theta + 2(|E_3|+|E_4|) \\
 &  =  & W^* - (6|E(H)| + 2z|V(H)|) +  2(3|E(H)| + z \cdot |V(H)|) \\
 &  = &  W^* \\
\end{array}
\]
This contradiction implies that condition (A) in Claim~A must be satisfied for $(X_{\one}^{opt},X_{\two}^{opt})$.
As $(X_{\one}^{opt},X_{\two}^{opt})$ is a Nash equilibrium, Claim~B implies that condition (B) in Claim~B holds.

Let $Y$ contain all vertices $u \in V(H)$ that satisfy $u' \in V' \cap X_{\one}$. By condition~(B) in Claim~B 
we note that $Y$ is a transversal in $H$. 
Let $(Y_{\one},Y_{\two})$ be the $G$-extension of $Y$, and note that, by Claim~D, $(Y_{\one},Y_{\two})$ is a Nash equilibrium satisfying the following,
where $0 \leq \epsilon_{Y} < 2z(1-2\pay{B})$:
\[
w(X_{\one}^{opt},X_{\two}^{opt}) \geq W(Y_{\one},Y_{\two}) = W^* + \epsilon_{Y} +|V(H)|\cdot z (2-2\pay{B})) - |Y| \cdot  2z(1-2\pay{B})
\]
This implies the following (as $\epsilon_{Y} < 2z(1-2\pay{B})$),
\[
\begin{array}{rcl} \vspace{0.2cm}
|T^{opt}| \; \leq \; |Y| & = & \frac{ W^* + \epsilon_{Y} + |V(H)|\cdot z (2-2\pay{B})) - w(X_{\one}^{opt},X_{\two}^{opt}) }{2z(1-2\pay{B})}  \\
& = & \left\lceil \frac{ W^* + |V(H)|\cdot z (2-2\pay{B})) - w(X_{\one}^{opt},X_{\two}^{opt}) }{2z(1-2\pay{B})} \right\rceil \\
\end{array}
\]

By \eqref{eq:asterix} this implies that the following holds, which completes the proof of Claim~E.
\[
|T^{opt}| = \left\lceil \frac{ W^* + |V(H)|\cdot z (2-2\pay{B})) - w(X_{\one}^{opt},X_{\two}^{opt}) }{2z(1-2\pay{B})} \right\rceil. \hfill\qedhere
\]
\end{proof}
This finishes the proof of the theorem.
\renewcommand\qedsymbol{$\square$}
\end{proof}

\newpage
\bibliographystyle{plainnat}
\bibliography{polyEcon.bib}


\end{document}